\documentclass[%
 jmp,%
 amsmath,amssymb,
preprint,11pt %
]{revtex4-1}

\usepackage{graphicx}
\usepackage{dcolumn}
\usepackage{bm}
\usepackage[hyperindex]{hyperref}
\usepackage{amsmath, amssymb, amsfonts, amsthm, amscd}
\usepackage{mathrsfs}

\newtheorem{theo}{Theorem}[section]
\newtheorem{prop}[theo]{Proposition}
\newtheorem{lemma}[theo]{Lemma}
\newtheorem{cor}[theo]{Corollary}
\newtheorem{rmk}[theo]{Remark}

\newcommand{\N} {\mathbb{N}}
\newcommand{\R} {\mathbb{R}}
\newcommand{\C} {\mathbb{C}}

\newcommand{\om}{\omega}
\newcommand{\f}{\frac}
\newcommand{\be}{\begin{equation}}
\newcommand{\n}{\noindent}
\newcommand{\ee}{\end{equation}}
\begin{document}



\title[Orbital and asymptotic stability for standing waves of a NLS equation with concentrated nonlinearity in dimension three]{Orbital and asymptotic stability for standing waves of a NLS equation with concentrated nonlinearity in dimension three}

\author{Riccardo Adami}\email{riccardo.adami@polito.it}\affiliation{Dipartimento di Scienze Matematiche, Politecnico di Torino, C.so Duca degli Abruzzi 24 10129 Torino, Italy}
\author{Diego Noja}\email{diego.noja@unimib.it}\affiliation{Dipartimento di Matematica e Applicazioni, Universit\`a di Milano Bicocca, Via Cozzi 53 20125 Milano, Italy}%
\author{Cecilia Ortoleva}\email{cecilia.ortoleva@math.cnrs.fr}
\affiliation{Laboratoire d'Analyse et de Math\'{e}matiques Appliqu\'{e}es, UMR 8050
Universit\'{e} Paris-Est, 
61 avenue du G\'{e}n\'{e}ral de Gaulle, 94010 Cr\'{e}teil Cedex, France 
}%

\date{\today}
\begin{abstract}We begin to study in this paper orbital and asymptotic stability of
standing waves for a model of Schr\"odinger equation with concentrated
nonlinearity in dimension three. The nonlinearity is obtained
considering a {point} (or contact) interaction with strength
$\alpha$, which consists of a singular perturbation of the Laplacian
described by a selfadjoint operator $H_{\alpha}$, and letting the
strength $\alpha$ depend on the wavefunction: $i\dot u= H_\alpha
u$, $\alpha=\alpha(u)$. 
It is well-known that the elements of the domain of such operator can
be written as the sum of a regular function and a function that
exhibits a singularity proportional to $|x - x_0|^{-1}$, where $x_0$ is
the location of the point interaction.
If $q$ is the so-called charge of the domain
element $u$, i.e. the coefficient of its singular part, then, in order
to introduce a nonlinearity, we let the strength
$\alpha$ depend on $u$ according to the law
$\alpha=-\nu|q|^\sigma$, with $\nu > 0$. 
This characterizes the model as a focusing
NLS with concentrated nonlinearity of power type. For such a model
we prove the existence of standing waves of the form $u (t)=e^{i\omega
  t}\Phi_{\omega}$, 
which are orbitally stable in the range $\sigma \in (0,1)$, and
orbitally unstable when $\sigma \geq 1.$   Moreover, we show that for $\sigma \in
(0,\frac{1}{\sqrt 2})$ every standing wave is asymptotically stable
in the following sense. Choosing
initial data close to the stationary
state in the energy norm, and belonging to a natural weighted $L^p$
space which allows 
dispersive estimates, the following resolution holds:
$u(t) = e^{i\omega_{\infty} t} \Phi_{\omega_{\infty}}
+U_t*\psi_{\infty} +r_{\infty}$, 
where $U$ is the free Schr\"odinger propagator,
$\omega_{\infty} > 0$ and $\psi_{\infty}$, $r_{\infty} \in L^2(\R^3)$
with 
$\| r_{\infty} \|_{L^2} = O(t^{-5/4}) \quad \textrm{as} \;\; t
\rightarrow +\infty$. 
Notice that in the present model the admitted nonlinearity for which
asymptotic stability of solitons is proved is subcritical, in the sense
that it does not give rise to blow up, regardless of the chosen
initial data. 
 
%
\end{abstract}

\maketitle

\section{Introduction} \label{introduzione}
In this paper we begin a systematic analysis of the stability of
solitary waves for a nonlinear Schr\"{o}dinger equation with a
nonlinearity concentrated in space dimension three.  In particular, we
show that the standing waves of the model are 
asymptotically stable in the sense that the evolution of the system in
a neighbourhood of a standing solution
admits a soliton resolution expansion: at large times, the evolution
decomposes as the sum of a standing wave (possibly with different
parameters from those of the reference initial soliton), a free linear
wave, and a small remainder with a spatial decay stronger than the
linear dispersive one.\par\noindent An analogous study concerning the
NLS equation with a concentrated nonlinearity in dimension one was
given in \cite{BKKS} and \cite{KKS}. These papers have been a source
of inspiration for the present work, in particular for what concerns
the general scheme of analysis and for some proofs. However, the one
and the three-dimensional models are different, in particular the
latter is strongly singular and its energy space is
not contained in $H^1(\R^3)$. This fact prevents us from following
step by step the techniques and the results of the cited papers;
in particular, no formal manipulations with delta distributions are
possible, and the full definition of a delta interaction as a point
perturbation of the Laplacian is needed in the analysis.  We shall
comment on that along the paper. \par\noindent We start by giving a
presentation of the model.  According to \cite{ADFT}, we construct a
Schr\"{o}dinger equation with concentrated nonlinearities in dimension
three by starting from the standard three-dimensional {\it linear}
Schr\"odinger operator with a so-called point or delta interaction
(\cite{Albeverio}). Point interactions are widely used in Quantum
Mechanics as models of contact or zero-range interactions and they are
intended to describe strongly concentrated potentials at a point. In
order to rigorously define a delta interaction located at the origin
of $\R^3$ we first consider the Laplacian restricted to the set
$C^{\infty}_0(\R^3\setminus \{0\})$ and obtain a symmetric non
selfadjoint operator with deficiency indices $(1,1)$. Second, by the
classical Von Neumann-Krejn theory there exists a one-parameter family
of selfadjoint extensions, which we denote by
$H_{\alpha}$. The operator $H_{\alpha}$ is defined on the
domain
\begin{equation*}
D(H_{\alpha}) = \{ u \in L^2(\R^3): \; u(x) = \phi(x) +q G_0(x) \,
\textrm{with} \; \phi \in L^2_{loc}(\R^3)\ , \nabla \phi \in L^2(\R^3)\ , \Delta \phi \in
L^2(\R^3),
\end{equation*}
\begin{equation}\label{dha} q \in \C,\quad \quad \lim_{x \rightarrow 0} (u(x) - q
G_0(x)) = \alpha q  \},
\end{equation}
\noindent where $G_0$ is the Green's function of the Laplacian in
three dimensions, i.e.
\begin{equation}	\label{Green}
G_0(x) = \frac{1}{4\pi |x|},
\end{equation}

\noindent and the action is given by $H_{\alpha} u(x) = -\Delta \phi(x)$, $x \in \R^3$.
To summarize, any element of the domain decomposes in a regular part
$\phi$ and a singular (Coulombian) part; the coefficient $q$ of the
singular part is conventionally called {\it charge}, and the boundary
condition imposes a relation between the charge and the value of the
regular part at the origin depending on the so-called {\it strength}
$\alpha$ of the point interaction, which is the parameter that fixes the selfadjoint extension. 
\par\noindent
An alternative equivalent and perhaps more direct construction, which
better justifies the interpretation and the physical meaning of
$H_{\alpha}$, can be given by defining $H_{\alpha}$ as a suitable scaling
limit (in norm resolvent sense) of a family of Schr\"odinger operators
of the form $-\Delta + V_{\epsilon}$, where $V_{\epsilon}$ is a short
range potential that approximates a delta distribution as $\epsilon\to
0$. A closer analysis of the above scaling procedure shows that the point interaction cannot be interpreted as a kind of "laplacian plus delta distribution", differently from the one dimensional case; moreover the parameter $\alpha$ appearing in the above definition and  characterizing the particular selfadjoint extension is related to zero energy
resonances of the approximating operators. For details and further
information see \cite{Albeverio}. 
\par\noindent Whatever the definition given to the operator
$H_{\alpha}$ is, we recall that, for $\alpha\geq 0$ (repulsive delta
interaction), $H_\alpha$ is positive and its spectrum is purely
absolutely continuous and coincides with $[0,+\infty)$, while for
$\alpha<0$ (attractive delta interaction) an isolated simple negative
eigenvalue $\lambda= -(4\pi\alpha)^2$ appears, corresponding to a
bound state. 
A second property relevant to the physical interpretation of the model
and related to the value of $\alpha$ is that the scattering length of
a delta interaction of strength $\alpha$ is given by
$-(4\pi\alpha)^{-1}\ $. 
The closed and lower bounded quadratic form  associated to $H_{\alpha}$ is
\begin{equation}	\label{quadratic}
\mathbf{H_{\alpha}}(u) = \int_{\R^3} |\nabla \phi x)|^2 dx
+\alpha |q|^2,
\end{equation}

\noindent defined on the domain of {\it finite energy states}
\begin{equation}    \label{def-V}
V = \{ u \in L^2(\R^3): \; u(x) = \phi(x) +q G_0(x), \, \textrm{with}
\; \phi \in L^2_{loc}(\R^3), \; \nabla \phi \in L^2(\R^3), \, q \in
\C \},
\end{equation}

\noindent which is a Hilbert space endowed with the norm
\begin{equation}	\label{norm-V}
\|u\|^2_V = \| \nabla \phi \|_{L^2} + |q|^2.
\end{equation}

\noindent Note that for a generic element $u$ of the form domain the charge $q$ and its regular part $\phi$ are independent of each other. 
Note also that the energy domain is strictly larger than $H^1(\R^3)$.
So, the linear problem cannot be considered as a small perturbation of
the standard free problem in the sense of the quadratic forms (at
variance with the one-dimensional case). An equivalent representation
of the energy space is obtained, fixed $\lambda>0$,  by
\begin{equation}\label{def-V2}
V = \left\{ u = \phi_{\lambda}+ q G_{\lambda}, \,
\textrm{with} \, \phi_{\lambda} \in H^1(\R^3), \, q \in \C\ , \ G_{\lambda}(x)=\frac{e^{-\lambda|x|}}{4\pi|x|} \right\},
\end{equation}

\noindent and one can define an equivalent energy norm by
$$\|u\|_V^2 = \| \nabla \phi_{\lambda} \|_{L^2}^2 +|q|^2, \quad \forall u \in
V.$$
Notice that $G_{\lambda}\in L^2(\R^3)$ and $\phi_{\lambda} \in H^1(\R^3)$, while in the representation \eqref{def-V} the regular part belongs 
to the {\it homogeneous} Sobolev space $D^1(\R^3)$ and $G_{0}\notin L^2(\R^3)$.
\par\noindent
Following \cite{ADFT}, the nonlinear model can be defined by allowing the strength
$\alpha$ to depend on $u$ as $\alpha(u) = -\nu |q|^{2\sigma}$, with $ \nu > 0,  \sigma > 0$, so that
$$D(H_{\alpha(u)}) = \{ u \in L^2(\R^3): \; u(x) = \phi(x) +q G_0(x) \,
\textrm{with} \; \phi \in H^2_{loc}(\R^3), \, \Delta \phi \in
L^2(\R^3),$$
$$q \in \C,\lim_{x \rightarrow 0} (u(x) - q
G_0(x)) = - \nu |q|^{2 \sigma} q  \},$$

\noindent and $H_{\alpha(u)} u = -\Delta \phi$.
\noindent
In the following sections, we often omit the notation $H_{\alpha(u)}$
in favour of $H_{\alpha}$ if no risk of confusion exists between the
linear and the nonlinear operator.  We stress that the nonlinarity we are
considering is {\em focusing}. It can be interpreted as modeling the
action of a defect in a  medium which exerts a nonlinear response to
the propagation.
We remark that a more general definition of concentrated
nonlinearities (with applications to the case of the wave equation) is
given in \cite{NP}. 

We consider the evolution
generated by the nonlinear operator $H_{\alpha(u)}$, i.e. 
\begin{equation}    \label{eq1}
i\frac{du}{dt} = H_{\alpha(u)} u.
\end{equation}

\noindent 
In the present literature, there is some physical and numerical
analysis of Schr\"odinger dynamics in the presence of nonlinear
defects,  mainly focused on the milder one-dimensional case (\cite{MA},\cite{SKBRC},\cite{DM}). 
The more technical construction of the three-dimensional problem 
has hindered extended modelistic study, numerical work as well as
rigorous analysis. Moreover, a certain amount of literature is devoted
to NLS with nonhomogeneous (i.e. $x$-dependent and decaying)
nonlinearities, yet with a relatively low decay at infinity
(see \cite{FW,GeS} and references therein).  

 \noindent 
Local (for any $\sigma > 0$) and global (for $\sigma < 1$) well-posedness of the Cauchy problem associated to
the nonlinear Schr\"{o}dinger equation \eqref{eq1} in the space $V$
have been established in \cite{ADFT} and \cite{ADFT2}. In
particular, \eqref{eq1} admits two conserved quantities called {\it
mass} and {\it energy}, defined as
$$\begin{array}{ll}
   M(u(t)) = \|u(t)\|_{L^2}^2, \qquad
   E(u(t)) =\frac{1}{2}\|\nabla \phi(t)\|_{L^2}^2 -\frac{\nu}{2\sigma +2}
|q(t)|^{2\sigma +2}.
\end{array}$$
\noindent 
In Section \ref{preliminari} we prove that equation \eqref{eq1} admits standing waves, i.e. solutions
of the form $u(x,t) = e^{i\omega t} \Phi_{\omega}(x)\ ,$ where the
profile or amplitude $\Phi_\omega$ up to a phase factor $e^{i\theta}$
is given by 
\begin{equation}	\label{fiom}
\Phi_{\omega}(x)=\left(\frac{\sqrt{\omega}}{4\pi \nu} \right)^{\frac{1}{2\sigma}}
\frac{e^{-\sqrt{\omega} |x|}}{4\pi |x|}\ .
\end{equation}
The set of standing waves is called the {\it solitary manifold}
$\mathcal{M}$, and the main concern of this paper consists in the
study of
the large-time evolution of initial data in the vicinity of $\mathcal{M}$.
A first result concerns stability and instability of standing
waves. 
Stability has to be intended as {\em orbital} stability, i.e. Lyapunov
stability up to symmetries of the equation, in this case up to gauge
($U(1)$) invariance. The orbit of $\Phi_{\omega}$ is then
${\mathscr
O}(\Phi_{\omega})=\{e^{i\theta}\Phi_{\omega}(x),\ \theta\in \R\}$. 
Thus, by definition, the
state $\Phi_{\omega}$ is orbitally stable if for every $\epsilon>0$
there exists $\delta>0$ such that 
\begin{equation*}
d(\psi(0), {\mathscr O}(\Phi_{\omega}))<\delta \quad \Rightarrow \quad d(\psi(t), {\mathscr O}(\Phi_{\omega}))<\epsilon \quad\quad \forall t > 0
\end{equation*}
where $d(\psi, {\mathscr O}(\Phi_{\omega}))=\inf_{u\in{\mathscr
O}(\Phi_{\omega})} \|\psi-u\|_{V}$.
A stationary state is said to be unstable if it is not stable. Then,
we have the following result,
proved in Section \ref{staborb}: \par\noindent
{\bf Theorem}  {\bf (Orbital Stability)}
Let us consider \eqref{eq1}. Then, for every $\omega>0$, 

\noindent (a) if $0<\sigma<1$, then the state $\Phi_{\omega}$ is orbitally stable 

\noindent (b) if $\sigma\geq 1$, then $\Phi_{\omega}$ is orbitally unstable.

\vskip3pt
\par\noindent

The result directly follows from Weinstein \cite{W1} and
Grillakis-Shatah-Strauss \cite{GSS} theory for the case $\sigma\neq
1$, while for the case $\sigma=1$ the pseudoconformal invariance of
the equation gives the instability by blow-up.\par\noindent 
The core of the paper is devoted to the study of the asymptotic
stability of the family of stationary states. Asymptotic stability
means, loosely speaking, that the solution $u(t)$ corresponding to an
initial datum $u(0)$ close to the family of orbits, approaches some
element of the family of orbits as $t\to \infty$. 
The analysis makes use 
of the representation
\begin{equation}\label{asrepr}
u(t,x) = e^{i\Theta(t)} \left( \Phi_{\omega(t)}(x) +\chi(t,x)
\right),
\end{equation}
\noindent where
$\Theta(t) = \int_0^t \omega(s) ds +\gamma(t),$ and $\gamma(t)$ is a
suitable phase. Namely, the solution is represented at every time as a
modulated solitary wave, with time dependent parameters, up to a fluctuating 
remainder $\chi$ which has to be controlled. Asymptotic stability of
the family of standing waves means that the modulating parameters
$\omega(t)$ and $\gamma(t)$ have a limit as $t \to \infty$, and the
fluctuation $\chi$ is in some sense a small and decaying dispersive
correction; the radiation damping through dispersion is responsible for the
``dissipative" asymptotic behaviour of the solution $u$ around the family of relative equilibria ${\mathscr
O}(\Phi_{\omega})\ .$ 
Notice that, however, in general the solution does not
converge to the solitary wave to which it was close initially.  
\par\noindent The subject of asymptotic stability of solitary
waves was pioneered by Soffer and Weinstein (\cite{SW1}, \cite{SW2}),
and Buslaev and Perelman (\cite{BP1}, \cite{BP2}), who developed the main
strategies and techniques, nowadays classical; a more
recent presentation is contained in \cite{BS}. Many relevant later
contributions refining and enlarging the hypotheses in the original
papers, as well as concerning the kind of initial admitted data and nonlinearities,
are contained in \cite{Cu,TY1,TY2,GNT,GS1,GS2,CM}. 
According to this consolidated analysis, one must preliminarily indagate the
spectrum of the linearization  of equation \eqref{eq1} around the
solitary solution. Writing $u=e^{i\omega t}(\Phi_\omega + R)$ and
identifying $R$ with the vector of its real and imaginary part, we
obtain that it satisfies the canonical system 
$$ 
J\frac{dR}{dt} =  \left[
  \begin{array}{cc}
    H_{\alpha_1} + \omega & 0 \\
    0 & H_{\alpha_2} + \omega \\
  \end{array}
\right]R\equiv DR
$$ 
\noindent
where $H_{\alpha_j}$ are (linear) delta interaction hamiltonian
operators  with fixed strength $\alpha_j$ that depend on the
stationary state $\Phi_\omega$ (through its charge) and on the
parameters of the model $\nu, \sigma$ (see eq. \eqref{alpha}). So the
dynamics of the linearization of the NLS around the standing wave
$\Phi_\omega$ is controlled by the nonselfadjoint (Hamiltonian) matrix
operator $L=JD$. The explicit characterization of the spectrum of the
linearization $L$ is possible due to the detailed knowledge of the properties of operators
$H_{\alpha_j}$. Such feature is unfrequent and
allows to avoid further spectral assumptions. 
The complete result is given in Section \ref{linearizzato},
Theorem \ref{ris}. Here it is sufficient to recall that    
in this paper we study asymptotic stability of standing waves in the
range $\sigma\in(0,1/\sqrt{2})$ only, which corresponds to $L$ having
no eigenvalues different from zero and no resonances at the threshold
of the essential spectrum.  A forthcoming paper will treat the case
$\sigma\in (1/\sqrt{2},1)$, where two simple eigenvalues $\pm i
2\sigma \sqrt{1 -\sigma^2} \omega$ appear. \par\noindent 
\par\noindent 
Let us notice that the representation \eqref{asrepr} amounts in fact
 to a change of coordinates from the original global $u$ to the new
 set $\{\omega, \gamma, \chi\}$, with a finite dimensional component
 given by $\{\omega,\gamma \}$, that describes the solitary manifold
 and an infinite dimensional one described by $\chi$. However, the
 representation is not unique, because any choice of
 ${\omega, \gamma}$ gives a corresponding choice of $\chi$ such that
 $u$ given by \eqref{asrepr} is a solution of \eqref{eq1}; so one has
 to restrict in some way the behaviour of the new parameters
 $\{\omega, \gamma, \chi\}$ of the solution.  To this end, we exploit the fact that the solitary manifold
 can be naturally endowed with a symplectic structure (see
 Section \ref{sw}) and it
 turns out that its tangent space $T_{\Phi_{\omega}}$  coincides with
 the generalized kernel of the linearization $L$. The generalized
 kernel is in turn non trivial, so the propagator $e^{-tL}$ has a
 component growing in time. A parametrization of the running
 approximate solitary wave in the neighborhood of the solitary
 manifold suitable for asymptotic analysis is hence obtained through a
 symplectic splitting in a component along the solitary manifold and a
 component transversal (symplectically orthogonal) to it. Requiring
 that the infinite dimensional component $\chi$ is purely transversal, i.e. projects to zero on the directions of the discrete spectrum, here reduced to the generalized kernel of the linearization,
 provides the set of the so called modulation (coupled) equations
 for the parameters $\omega(t)$ and $\gamma(t)$, as well as a corresponding
 partial differential equation for $\chi$ (see \cite{FGJS1} for an
 enlightening description of the symplectic projection method). The
 goal is to establish the asymptotic behaviour of the solutions to the
 modulation equations with
 a simultaneous control of the decay of the nonlinear part $\chi$,
 through the so-called majorant's method
 (see \cite{BP1,BP2,BS}).\par\noindent 
The main result of this paper is the following, and it is proven in
 Section  \ref{stabasi}. \par\noindent
{\bf Theorem} ({\bf Asymptotic stability})

Assume $\sigma \in (0, 1/\sqrt{2})$. Let $u \in C(\R^+, V)$ be a
solution to equation \eqref{eq1} with $u(0) = u_0 \in V \cap L^1_w$
and denote $d = \| u_0 -e^{i\theta_0} \Phi_{\omega_0} \|_{V \cap
L^1_w},$ for some $\omega_0 > 0$ and $\theta_0 \in \R$. Then, if $d$
is sufficiently small, the solution $u(t)$ can be decomposed as
follows
\begin{equation}
u(t) = e^{i\omega_{\infty} t} \Phi_{\omega_{\infty}}
+U_t*\psi_{\infty} +r_{\infty}(t), 
\end{equation}

\noindent where $\omega_{\infty} > 0$ and $\psi_{\infty}$,
$r_{\infty}(t) \in L^2(\R^3)$, with $\| r_{\infty}(t) \|_{L^2} =
O(t^{-5/4})$ as $t \rightarrow +\infty.$
 
\vskip5pt
\par\noindent
In the previous statement, $L^1_w$ is defined in Section \ref{dispersive} and
is a weighted space of integrable functions. The weight guarantees the
validity of the dispersive estimates needed in order to control the decay of
the transversal evolution, and it seems at present unavoidable in view
of the singularity of finite energy states. Moreover, it imposes a
certain localization on the the admitted initial data, which seems to
be a technical requirement. The norm $||\cdot||_{V\cap L^1_{w}}$ is defined as the maximum of the norms of the two Banach spaces $V$ and $L^1_w$.
\par\noindent
Concerning the treatment of the modulation equations, one of the
main additional difficulties with respect to standard models, and in
particular with the case of concentrated nonlinearities in one
dimension treated in \cite{BKKS} and \cite{KKS}, is that the equations
controlling the evolution of the transversal part $\chi$ have domains
that change with time. This fact forced us to make use of the
variational formulation (i.e. in terms of quadratic forms) instead of
the traditional strong formulation (i.e. in terms of operators and
their domains). The same problem propagates to the proof of the
asymptotics given in the above theorem.  A last remark concerns the
seemingly anomalous value of the nonlinearities where asymptotic
stability is proven; this because in the typical situations, when
standard NLS with or without potential is treated, it is difficult to
have information about subcritical nonlinearities (but see the notably
exception in \cite{KM}), and in particular pure power. On the other
hand, the present model 
corresponds to an inhomogeneous (space dependent and strongly singular)
nonlinearity; this seems to indicate that the analysis of specific
models can give results not accessible to general theory, at least at present.
The paper is organized as follows. In Section II we fix some
notation, describe the set of standing waves for the system of
interest, and deduce the linearized evolution around a standing
solution. In Section III, using the method by Grillakis, Shatah and
Strauss, we prove the orbital stability for any standing wave in the
case of low nonlinearity power (i.e. $\sigma < 1$, and orbital
instability in the case of large nonlinearity power (i.e. $\sigma >
1$); in the critical case ($\sigma = 1$) we directly show that any
stationary state is affected by instability due to the vicinity of
initial data that give rise to blow-up solutions. Section IV is
devoted to the study of the linearized problem: first, the resolvent
of the linearized generator of the evolution is explicitly
constructed; then, it is used in order to derive dispersive estimates
in suitable weighted spaces. In Section V we start the analysis of the asymptotic completeness by deducing the modulation
equations of the system. In Section VI we prove the
decay rate in time of the solutions to the modulation equations.
Finally. in Section VII we prove the result about asymptotic stability.
The paper ends by three appendices related to the content of Section
IV. In the first appendix we construct the generalized kernel of the
generator of the linearized evolution, in the second we prove a
convenient expression for the resolvent of the same operator, while in
the last appendix we give the explicit linearized dynamics along the generalized kernel.

\vskip10pt
\par\noindent
{\bf Acknowledgments} The authors are grateful to Gianfausto
Dell'Antonio and Galina Perelman for several discussions. 
R. A. is partially supported by 
 the PRIN2009 grant {\em Critical Point Theory and Perturbative 
Methods for Nonlinear Differential Equations}.

\noindent
This paper is part of the Ph.D. thesis of C. O.




\section{Preliminaries} \label{preliminari}

\subsection{Hamiltonian structure}\label{sw}
We consider $L^2(\R^3,\C)$ as a real Hilbert space endowed with the scalar product
\begin{equation}\label{scalarproduct}
(u,v)_{L^2} = \Re \int_{\R^3} u \overline{v} \, dx = \int_{\R^3} (\Re v \Re u +\Im v \Im u) dx.
\end{equation}
\noindent It is sometimes convenient to shift from the complex valued
representation of $u$ to the vector real valued one through the
identification $u=\Re u + i \Im u \mapsto (\Re u, \Im
u)=(u_1,u_2)$. As a consequence, $H^s(\R^3, \C) \cong H^s(\R^3,
\R^2)$, while multiplication by $i$ is equivalent to multipication by
the matrix $-J$, where  
\begin{equation}\label{J}
J = \left[
  \begin{array}{cc}
    0 & 1 \\
    -1 & 0 \\
  \end{array}
\right]\ .
\end{equation}

\noindent The space $L^2(\R^3)$ is also a symplectic manifold  when endowed with the symplectic form
\begin{equation}	\label{formasimplettica}
\Omega(u,v)=\Im \int_{\R^3} u{\overline v} \ dx = \int_{\R^3} (\Re v \Im u -\Im v \Re u) dx =\int_{\R^3} (u_2 v_1 -u_1 v_2) dx.
\end{equation}

\noindent Along the paper we often shift between real and complex
representation when no ambiguity occurs. 
\par\noindent
In our model the Hamiltonian functional coincides with the total energy and  (exploiting the decomposition of an element of the form domain in regular and singular part) it is given by
\begin{equation}	\label{energy}
 E(u) =\frac{1}{2}\|\nabla \phi\|_{L^2}^2 -\frac{\nu}{2\sigma +2}
|q|^{2\sigma +2}, \ \ u=\phi+qG_0 \in V.
\end{equation}

\par\noindent
Correspondingly, the NLS \eqref{eq1} takes the hamiltonian form
\begin{equation}\label{eqham}
\frac{du}{dt}=J \ E^{'} (u)\ .
\end{equation} 
where the prime denotes the differential of the considered functional at the point $u$.
\subsection{Standing waves}
Standing waves are solutions of the equation \eqref{eq1} of the form $u(x,t) = e^{i\omega t} \Phi_{\omega}(x).$
It immediately follows that if a standing wave exists, then the
amplitude $\Phi_{\omega}$ satisfies the nonlinear equation 
\begin{equation}    \label{gs_eq}
 H_{\alpha(\Phi_{\omega})} \Phi_{\omega} +\omega \Phi_{\omega} = 0.
\end{equation}

\begin{prop}
Standing waves for equation \eqref{eq1} exist if and only if $\nu > 0$. In
such a case the set of solitary waves is given by the two-dimensional manifold
\begin{equation}	\label{M}
\mathcal{M} \: = \ 
\left\{\ e^{i\Theta}\ \Phi_{\omega}
\ , \ \omega > 0\ ,\ \Theta \in [0, 2 \pi) \right\},
\end{equation}

\noindent where the function $$\Phi_\om(x)=\left(\frac{\sqrt{\omega}}{4\pi \nu} \right)^{\frac{1}{2\sigma}}
\frac{e^{-\sqrt{\omega} |x|}}{4\pi |x|}\ $$ and the parameters $\omega$ and $\Theta$ play the role of local coordinates.
\end{prop}

\begin{proof}

\noindent Recall that the function
$G_0$ defined in \eqref{Green} satisfies the equation
$-\triangle G_0 = \delta$

\noindent where $\delta$ is the Dirac's delta distribution centred at
$x = 0$. Hence, for $x \neq 0$ equation \eqref{gs_eq} is
equivalent to $-\triangle \Phi_{\omega}(x) +\omega \Phi_{\omega}(x ) = 0.$
Let us introduce, with a slight abuse, the function $f(r,\theta, \phi)=\Phi_{\omega}(x)$ and consider the corresponding equation in spherical
coordinates, namely
$$-\frac{\partial^2 f}{\partial r^2} -\frac{2}{r} \frac{\partial
f}{\partial r} -\frac{1}{r^2} \frac{\partial^2 f}{\partial \phi^2}
-\frac{\cos \phi}{r^2 \sin \phi} \frac{\partial f}{\partial \phi}
-\frac{1}{r^2 \sin^2 \phi} \frac{\partial^2 f}{\partial \theta^2}
+\omega f = 0,$$

\noindent and exploit the spherical harmonics expansion of the
solution $f(r, \theta, \phi) = \sum_{l=0}^{+\infty} \sum_{j = -l}^l v_{l, j}(r)
Y_{l,j}(\theta, \phi),$

\noindent where $Y_{l,j}$ denotes the set of spherical harmonics
which is an orthonormal basis of $L^2([0, \pi] \times [0, 2\pi], \sin \theta d\theta d\phi)$.
Since
$$\frac{\partial^2 Y_{l,j}}{\partial \phi^2} +\frac{\cos
\phi}{\sin \phi} \frac{\partial Y_{l,j}}{\partial \phi}
+\frac{1}{\sin^2 \phi} \frac{\partial^2 Y_{l,j}}{\partial
\theta^2} = -\lambda Y_{l,j}, \qquad \textrm{for some} \;\; \lambda
\in \C,$$
one has that $\lambda$ belongs to the set $\{ \lambda_l : = l (l+1),
\, l \in {\mathbb N} \}$, and so 
the functions $v_{l,j}$ solve $-v_{l, j}''(r) -\frac{2}{r} v_{l, j}'(r) +\left( \omega
-\frac{\lambda_l}{r^2} \right) v_{l,j}(r) = 0.$
Then, from formula
8.491.6 in \cite{tavole},
$$v_{j,l}(r) = \frac{1}{\sqrt{r}} Z_{\sqrt{\frac{1}{4} +\lambda}}
(\sqrt{\omega} r),$$
where $Z_\nu$ is a Bessel's function.
By the asymptotic expansions 8.443 and 8.451.1 in \cite{tavole} one
immediately has that if $\lambda \neq 0$, then $v_{j,l}$ cannot belong
to $L^2(\R^+, r^2 dr)$. 
Hence, we fix $\lambda = 0$ and denote $\Phi_{\omega}(x) =
\frac{v(r)}{r}$, $r = |x|$. Thus
$v$ has to be a square-integrable solution of $v''(r) -\omega v(r) = 0,$ and finally
$$\Phi_{\omega}(x) = \frac{q e^{-\sqrt{\omega}
|x|}}{4\pi |x|},$$

\noindent for some $q \in \C$ and $\omega > 0$. Furthermore, by the boundary  condition in the definition of $D
(H_{\alpha(u)})$,
i.e. $\lim_{x \rightarrow 0} (\Phi_{\omega}(x) - q G_0(x))
=\alpha(\Phi_{\omega}) q$, one gets $-\frac{q}{4\pi} \sqrt{\omega} =
-\nu |q|^{2\sigma} q$, and supposing $\nu\neq 0$ one obtains
$|q|^{2\sigma} = \frac{\sqrt{\omega}}{4\pi \nu}.$ 
 This requires
$\nu > 0$, so
$$\Phi_{\omega}(x) = \left( \frac{\sqrt{\omega}}{4\pi \nu}
\right)^{\frac{1}{2\sigma}} \frac{e^{-\sqrt{\omega} |x|}}{4\pi
|x|}$$

\noindent which, up to a phase factor, gives the stated result.
In the case $\nu=0$, from boundary condition we get $q=0$ or
$\omega=0$. If $q=0$ , then the function $u$ vanishes. If $\omega=0$,
then  one has $u(x)=\frac{1}{4\pi |x|}$, which is the resonance function of the delta interaction with vanishing strength, but it is not an element of the operator domain, and it does not solve the stationary equation \eqref{gs_eq}. So for $\nu=0$ standing waves do not exist.
\end{proof}

\noindent In the following we denote
$q_{\omega} = \left( \frac{\sqrt{\omega}}{4\pi \nu}
\right)^{\frac{1}{2\sigma}}.$

\subsection{Linearization of $H_{\alpha(u)}$ around
  $\Phi_\om$} \label{lineariz}
The linearization of equation \eqref{eq1} around a
stationary solution is not completely obvious, due to the fact that
the nonlinearity is embodied in the domain of the operator $H_{\alpha(u)}$ {\it and not}
in the action of the operator itself. Nevertheless, we can
consider the Hamiltonian associated to equation \eqref{eq1} given by
formula \eqref{eqham} and notice that the nonlinearity no longer
appears in the domain $V$ but directly in the Hamiltonian
functional. So we derive the linear operator which approximates
$H_{\alpha(u)}$ from the quadratic form which approximates $E(\Phi_\om)$ and
obtain the following result.
\begin{prop}	\label{E''}
The Hessian ${E''}(\Phi_\omega)$ of the functional $E$ can be 
represented as ${E''}(\Phi_\omega)(h,k)=\langle H_{\alpha,Lin}\ h, k
\rangle $, where  $H_{\alpha,Lin}$ is the linear operator given by 
$$H_{\alpha,Lin} = \left[
  \begin{array}{cc}
    H_{\alpha_1} & 0 \\
    0 & H_{\alpha_2} \\
  \end{array}
\right],$$

\noindent where the operators $H_{\alpha_1}$ and  $H_{\alpha_1}$ are the standard point interactions defined in the Introduction (see \eqref{dha}) and the fixed parameters $\alpha_1$ and $\alpha_2$ are given by
\begin{equation}    \label{alpha}
    \alpha_1  = -\nu (2\sigma +1) |q_{\omega}|^{2\sigma} = - \f{2
    \sigma + 1}{4 \pi} \sqrt \omega, \qquad
    \alpha_2 = -\nu |q_{\omega}|^{2\sigma} = - \f {\sqrt \omega} {4 \pi}.
\end{equation}
$H_{\alpha,lin}$ is selfadjoint with respect to the real scalar product in $L^2(\R^3,\C)\ .$
\end{prop}

\begin{proof}
\noindent 
The first G{\^{a}}teaux derivative of $E(u)$ reads
\begin{equation}\label{Gateaux1}
E'(u)[h]={\frac{d}{d\epsilon}} \{E(u+\epsilon h)\}_{\epsilon=0}=
\Re \int_{\R^3} \nabla \phi_u(x)\cdot \overline{\nabla \phi_h (x)} dx -
\nu |q_u|^{2\sigma}\Re(q_u 
\overline{q_h})\ \ \ \forall u,h \in V,
\end{equation}
while the second G{\^{a}}teaux derivative at $\Phi_\omega$ reads
\begin{equation*}
\frac{\partial^2}{{\partial\epsilon}{\partial \lambda }}
\{{E}(\Phi_{\omega}+\epsilon h +\lambda k)\}_{\epsilon=0, \lambda=0} =
\Re \int_{\R^3} \nabla \phi_h (x)\cdot \overline{\nabla \phi_k (x)} dx -
\frac{\partial^2}{{\partial\epsilon}{\partial \lambda }}
\{\frac{\nu}{2\sigma +2} |q_{\Phi_{\omega}+\epsilon h +\lambda k}|^{2\sigma +2}\}_{\epsilon=0,
  \lambda=0}\ . 
\end{equation*}
\par\noindent
The last term gives, after some calculation, the contribution (here $h=(h_1,h_2),\ k=(k_1,k_2)$)
$$
\frac{\partial^2}{{\partial\epsilon}{\partial \lambda }}
\{- \frac{\nu}{2\sigma +2} |q_{\Phi_{\omega}+\epsilon h +\lambda k}|^{2\sigma +2}\}_{\epsilon=0,
  \lambda=0}= 
- \nu|q_{\omega}|^{2\sigma}[(2\sigma+1)q_{h_1}q_{k_1} + q_{h_2}q_{k_2}]\ .
$$
So $E^{\prime \prime}(\Phi_\om)$ is given by the direct sum of two
quadratic forms: one is acting on the real part of the functions $h$
and $k$, and the other on the imaginary part. The term related to the
real
part is a lower bounded quadratic form whose corresponding
selfadjoint operator is $H_{\alpha_1}$, while the quadratic form
related to the imaginary part corresponds to the operator
$H_{\alpha_2}$ ($\alpha_1$ and $\alpha_2$ have been defined in 
\eqref{alpha}). Then, the operator $H_{\alpha, Lin}$ representing
the entire quadratic form $E^{\prime \prime}(\Phi_\om)$ is
self-adjoint and the proof is complete.
\end{proof}
Now, to get the linearized equation set $u(t)=e^{i\omega t}(\Phi_{\omega} +R)$ and obtain
$$
\frac{d}{dt} R =J (E'(\Phi_{\omega}) + \omega\Phi_{\omega}) + J(E^{''}(\Phi_{\omega}) + \omega)R + {\rm higher\ order\ terms} \simeq J (H_{\alpha,Lin}+\omega) R\ .
$$
\noindent Summing up, the linearized equation \eqref{eq1} becomes
\begin{equation}    \label{eq_linearizzata0}
\frac{dR}{dt} = J D R,
\end{equation}

\noindent where $\displaystyle D = \left[
  \begin{array}{cc}
    L_1 & 0 \\
    0 & L_2 \\
  \end{array}
\right],$ with 
\begin{equation}\label{elleunoedue}
L_j = H_{\alpha_j} +\omega\ , \qquad j = 1, 2\ .
\end{equation}
 Notice that the operator 
\begin{equation}\label{venti}
JD :=L =  \left[
  \begin{array}{cc}
    0 & L_2 \\
    -L_1 & 0 \\
  \end{array}
\right],
\end{equation}
is not selfadjoint nor skew adjoint. Nevertheless, a standard
application of Hille-Yosida theorem and a simple analysis of the
resolvent of $L$ which takes into account the factorized structure
$L=JD$ with $D$ s.a. shows that it generates a semigroup of linear
operators with (at most) exponential growth in time. A more precise
analysis of the resolvent of the operator $L$ will be given in Theorem
\ref{ris} and in the appendix \ref{propagator} we will prove that the
semigroup has in fact a linear growth (see Theorem \ref{stima-disc})
in the case here interesting, i.e. $\sigma\in (0,1/\sqrt{2}). $
\par\noindent

\section{Orbital stability}\label{staborb}
In order to prove the orbital stability of the stationary solutions
to equation \eqref{eq1}, we apply Grillakis-Shatah-Strauss theory, and
in particular Theorem 2 in \cite{GSS}. As a first step, we recall the
following known fact proved 
in \cite{Albeverio}.

\begin{prop}
If $\alpha(u) = \alpha$ where $\alpha < 0$ is a constant, then
\begin{equation}    \label{spettro}
\sigma(H_{\alpha}) \equiv \{ -(4\pi \alpha)^2 \} \cup [0, +\infty).
\end{equation}
\end{prop}

\noindent Thanks to the last proposition one can prove the following
lemma which implies the spectral properties needed to verify Assumption 3 in \cite{GSS}.

\begin{lemma}   \label{ass3}
The spectrum of the operator $D$ is
$$\sigma(D) = \{ -4\sigma (\sigma +1) \omega, 0 \} \cup [\omega,
+\infty),$$

\noindent and $\displaystyle \ker(D) = \emph{\textrm{span}} \left\{ \left(
                   \begin{array}{ll}
                     0\\
                     \Phi_{\omega}
                   \end{array}
                 \right) \right\}.$

\end{lemma}

\begin{proof}
Since $D$ is the direct sum of the operators $L_1$ and $L_2$ acting on
$L^2(\R^3) \oplus L^2(\R^3)$, its spectrum is given by the union of
$\sigma(L_1)$ and $\sigma(L_2)$. From \eqref{spettro} follows 
$$\sigma(H_{\alpha_1}) = \{ -(2\sigma +1)^2 \omega \} \cup [0,
+\infty), \quad \sigma(H_{\alpha_2}) = \{ -\omega \} \cup [0, +\infty).$$

\noindent Then
$$\sigma(L_1) = \sigma(H_{\alpha_1}) +\omega = \{ -4\sigma (\sigma
+1) \omega \} \cup [\omega, +\infty), \quad \sigma(L_2) = \sigma(H_{\alpha_2}) +\omega = \{ 0 \} \cup [\omega,
+\infty).$$

\noindent Hence, $\ker(L_1) = \{ 0 \}$ and $\ker(L_2) = \textrm{span} \{
\Phi_{\omega} \}$, which concludes the proof.
\end{proof}

\noindent We can now prove the following

\begin{theo}{\textbf{\emph{(Orbital stability)}}}
For each $\omega > 0$, if $0 < \sigma < 1$, then $\Phi_{\omega}$ is
orbitally stable. If $\sigma > 1$, then $\Phi_{\omega}$ is orbitally
unstable.
\end{theo}

\begin{proof}
\noindent Well-posedness and existence of a branch of standing waves, i.e. Assumptions 1 and 2 in \cite{GSS}, are proved in
\cite{ADFT} and \cite{ADFT2} and in the previous section, while Assumption 3 is true thanks to
Lemma \ref{ass3}. Hence, from Theorem 3 in \cite{GSS} we have orbital
stability if  $\frac d {d\omega} \| \Phi_\omega
\|^2_{L^2(\R^3)}
> 0$ and orbital instability if
$\frac d {d\omega} \| \Phi_\omega
\|^2_{L^2(\R^3)}
<0$. In order to inspect the sign of $\frac d {d\omega} \| \Phi_\omega
\|^2_{L^2(\R^3)}
$, we compute
$$\| \Phi_{\omega}\|_{L^2}^2 = \left( \frac{\sqrt{\omega}}{4\pi \nu}
\right)^{\frac{1}{\sigma}} \frac{1}{8\pi \sqrt{\omega}},$$

\noindent hence
$\frac d {d\omega} \| \Phi_\omega
\|^2_{L^2(\R^3)}= \frac{1}{8\pi (4\pi \nu)^{1/\sigma}} \frac{1
-\sigma}{2\sigma} \omega^{\frac{1 -3\sigma}{2\sigma}},$
which concludes the proof.
\end{proof}

\subsection{The case $\sigma = 1$}
Since Theorem 3 in \cite{GSS} does not give information about
orbital stability of the stationary state $e^{i\omega
t}\Phi_{\omega}$ when $\frac d {d\omega} \| \Phi_\omega
\|^2_{L^2(\R^3)}= 0$, we need to inspect
the case $\sigma = 1$ apart.
In such case, equation \eqref{eq1} exhibits one additional symmetry (see \cite{ADFT2}).
\begin{rmk}
Equation \eqref{eq1} is invariant under the pseudoconformal
transformation
$$u_{pc}^T(t,x) = \frac{e^{-i\frac{|x|^2}{4(T-t)}}}{(T-t)^{3/2}}
u\left(\frac{1}{T-t}, \frac{|x|}{T-t}\right).$$

\end{rmk}

\noindent In \cite{ADFT} it is proved that equation \eqref{eq1} may
have some non global solutions which blow up, in the following sense:
the solution $u(t)$ of equation \eqref{eq1} blows up (in the future) at
time $T <+\infty$ if
$$\limsup_{t \rightarrow T^-} \| \nabla \phi_u \|_{L^2} = +\infty.$$
Here $\phi_u$ is the regular part of the function $u$ and $q_u$ is the corresponding charge
according to the decomposition in \eqref{def-V}.
\noindent Due to energy conservation this condition is equivalent to $\limsup_{t \rightarrow T^-} |q_u(t)| = +\infty$.


\noindent Thanks to the pseudoconformal  invariance we 
prove that in any neighbourhood (in energy norm) of each
standing wave there are initial data of a blow up solution.
\begin{theo}
Fix $\sigma = 1$ and $\omega > 0$. For any $\delta > 0$ there exists a
blow up solution $u(t) 
\in V$ such that $\|u(0) -\Phi_{\omega}\|_V < \delta$.

\end{theo}

\begin{proof}
Applying the pseudoconformal transformation to the solitary wave
$e^{i\widetilde{\omega} t}\Phi_{\widetilde{\omega}}$ one gets that for any $T > 0$, the function
$$u_{\widetilde{\omega},T}(t,x) = e^{i\frac{\widetilde{\omega}}{T -t}} \frac{\widetilde{\omega}^{1/4}}{\sqrt{4\pi
\nu}} \frac{e^{-\frac{\sqrt{\widetilde{\omega}} |x|}{T -t}}}{4\pi \sqrt{T
-t} |x|} e^{-i\frac{|x|^2}{4(T-t)}}$$

\noindent is a solution to equation \eqref{eq1}. Thus, for any $T > 0$, the initial datum $\displaystyle u_T(x) = e^{i\frac{\widetilde{\omega}}{T}} \frac{\widetilde{\omega}^{1/4}}{\sqrt{4\pi
\nu}} \frac{e^{-\frac{\sqrt{\widetilde{\omega}} |x|}{T}}}{4\pi \sqrt{T} |x|} e^{-i\frac{|x|^2}{4T}}$ gives rise to a solution that blows up at time $T$. Now, let $\widetilde{\omega}$ depend on $T$ as $\widetilde{\omega} = \omega T^2$, so that $u_T(x) = e^{-i\frac{|x|^2}{4T}} \Phi_{\omega}(x)$.

\noindent We prove the theorem by showing that $\| (
e^{-i\frac{|\cdot|^2}{4T}} -1 ) \Phi_{\omega} \|_V \rightarrow 0$ as
$T \rightarrow +\infty$. Indeed, noting that the function $(
e^{-i\frac{|\cdot|^2}{4T}} -1 ) \Phi_{\omega}$ belongs to $H^1(\R^3)$,
$$\| (
e^{-i\frac{|\cdot|^2}{4T}} -1 ) \Phi_{\omega} \|_V = 
\| \nabla (( e^{-i\frac{|\cdot|^2}{4T}} -1
) \Phi_{\omega}) \|_{L^2} \leq \frac{1}{2T} \|
|\cdot| \Phi_{\omega} \|_{L^2} +\frac{1}{4T} \|
|\cdot|^2 \nabla \Phi_{\omega} \|_{L^2} \rightarrow 0, \quad
T \rightarrow +\infty.$$ 

\end{proof}

\section{Spectral and dispersive properties of linearization $L$}
\label{linearizzato}
Here we study the long time behaviour of equation
\eqref{eq_linearizzata0},
that is the linearization of
 \eqref{eq1}
around the stationary solution $e^{i\omega t} \Phi_{\omega}$.

\noindent The generalized kernel of the operator $L$ (see
\eqref{venti}) is defined as
$\displaystyle
  N_g(L) = \bigcup_{k \in \N} \ker (L^k)$.

\noindent
In what follows let us denote
$$\begin{array}{ll}
    \varphi_{\omega}(x) = \frac{d \Phi_{\omega}}{d\omega}(x) = \frac{1}{4\sigma
\omega} \left( \frac{\sqrt{\omega}}{4\pi \nu} \right)^{\frac{1}{2\sigma}}
\frac{e^{-\sqrt{\omega} |x|}}{4\pi |x|} -\frac{1}{2\sqrt{\omega}} \left(
\frac{\sqrt{\omega}}{4\pi \nu} \right)^{\frac{1}{2\sigma}}
\frac{e^{-\sqrt{\omega} |x|}}{4\pi},\\
    g_{\omega}(x) = \frac{\omega^{\frac{1}{4}}}{\sqrt{4\pi \nu}} |x|
\frac{e^{-\sqrt{\omega} |x|}}{4\pi}, &  \\
    h_{\omega}(x) = \frac{\omega^{\frac{1}{4}}}{\sqrt{4\pi \nu}} \left(
-\frac{1}{4\omega^{\frac{3}{2}}} \frac{e^{-\sqrt{\omega} |x|}}{4\pi |x|}
+\frac{1}{2\omega} \frac{e^{-\sqrt{\omega} |x|}}{4\pi} +\frac{1}{2\sqrt{\omega}}
|x| \frac{e^{-\sqrt{\omega} |x|}}{4\pi} +\frac{1}{3} |x|^2
\frac{e^{-\sqrt{\omega} |x|}}{4\pi} \right).
\end{array}$$

\noindent In Appendix \ref{nucleo} we prove the following theorem.

\begin{theo}	\label{kergen}
If the nonlinearity power $\sigma$ is different from $1$, then $N_g(L)
= \emph{\textrm{span}} \left\{ \left(
                   \begin{array}{ll}
                     0\\
                     \Phi_{\omega}
                   \end{array}
                 \right),
                \left(
                   \begin{array}{ll}
                     \varphi_{\omega}\\
                     0
                   \end{array}
                 \right)
 \right\}.$

\noindent Moreover, if $\sigma = 1$, then $N_g(L) =
\emph{\textrm{span}} \left\{ \left(
                   \begin{array}{ll}
                     0\\
                     \Phi_{\omega}
                   \end{array}
                 \right),
                \left(
                   \begin{array}{ll}
                     \varphi_{\omega}\\
                     0
                   \end{array}
                 \right),
                \left(
                   \begin{array}{ll}
                     0\\
                     g_{\omega}
                   \end{array}
                 \right),
                \left(
                   \begin{array}{ll}
                     h_{\omega}\\
                     0
                   \end{array}
                 \right)
 \right\}.$

\end{theo}

\par\noindent
In the following section we provide an explicit description of the
spectrum of the non-selfadjoint operator $L$ and the dispersive
estimates for the action of the propagator $e^{-Lt}$ upon the
absolutely continuous subspace.

\subsection{The resolvent and the spectrum of the linearized operator}
The purpose of this section is to prove
an explicit formula for the resolvent of the linearized
operator. For later convenience we denote
\be \label{gom}
G_{\omega \pm i\lambda}(x) = \frac{e^{i \sqrt{-\omega \mp
i\lambda} |x|}}{4\pi |x|} \qquad \omega > 0, \lambda \in \C,
\ee

\noindent with the prescription $\Im{\sqrt{-\omega \pm i\lambda}} >
0$.

\noindent Furthermore, we make use of the notation $\langle g,h \rangle : =
\int_{\R3} g(x) h(x) \, dx$.

\noindent
We prove the following

\begin{theo}    \label{ris}
The resolvent $R(\lambda) = (L -\lambda I)^{-1}$ of the operator $L$
defined in \eqref{venti} is given by
\begin{equation}    \label{eq:risolvente}
R(\lambda) = \left[
  \begin{array}{cc}
    -\lambda \mathcal{G}_{\lambda^2} * & -\Gamma_{\lambda^2} * \\
    \Gamma_{\lambda^2} * & -\lambda \mathcal{G}_{\lambda^2} * \\
  \end{array}
\right] +\frac{4\pi}{W(\lambda^2)} i \left[
  \begin{array}{cc}
    \Lambda_1 & i\Sigma_2 \\
    -i\Sigma_1 & \Lambda_2\\
  \end{array}
\right],
\end{equation}
where
$$W (\lambda^2) = 32\pi2 \alpha_1 \alpha_2 -4i \pi (\alpha_1 +\alpha_2)
\left(\sqrt{-\omega +i\lambda} +\sqrt{-\omega -i \lambda} \right)
-2\sqrt{-\omega +i\lambda} \sqrt{-\omega -i\lambda},$$
and formula \eqref{eq:risolvente} holds
for all $\lambda \in \C \setminus \{ \lambda \in \C: \;
W(\lambda^2) = 0, \;\; \textrm{or} \;\; \Re(\lambda) = 0 \;
\textrm{and} \; |\Im(\lambda)| \geq \omega \}$. Furthermore, the
symbol $*$ in  \eqref{eq:risolvente}
denotes the convolution and
$$
\mathcal{G}_{\lambda^2}(x) = \frac{1}{2i \lambda} \left(
G_{\omega -i\lambda}(x) -G_{\omega +i\lambda}(x) \right), \quad 
\Gamma_{\lambda^2}(x) = \frac{1}{2} \left( G_{\omega
-i\lambda}(x) +G_{\omega +i\lambda}(x) \right).
$$
Finally, the entries of the second matrix are finite rank operators
whose action on $f \in L^2 (\R^3)$ reads
\begin{equation} \label{lambdasigma}
\Lambda_1 f = [i\lambda (4\pi \alpha_2 -i\sqrt{-\omega +i\lambda})
\langle\mathcal{G}_{\lambda^2}, f \rangle -(4\pi \alpha_1 -i\sqrt{-\omega
+i\lambda}) \langle\Gamma_{\lambda^2}, f \rangle]G_{\omega +i\lambda}
+
\end{equation}
$$+[i\lambda (4\pi \alpha_2 -i\sqrt{-\omega -i\lambda})
\langle \mathcal{G}_{\lambda^2}, f \rangle +(4\pi \alpha_1 -i\sqrt{-\omega
-i\lambda}) \langle \Gamma_{\lambda^2}, f \rangle]G_{\omega -i\lambda},$$
$$\Lambda_2 f = [i\lambda (4\pi \alpha_1 -i\sqrt{-\omega +i\lambda})
\langle \mathcal{G}_{\lambda^2}, f \rangle -(4\pi \alpha_2 -i\sqrt{-\omega
+i\lambda}) \langle \Gamma_{\lambda^2}, f \rangle]G_{\omega +i\lambda} +$$
$$+[i\lambda (4\pi \alpha_1 -i\sqrt{-\omega -i\lambda})
\langle \mathcal{G}_{\lambda^2}, f \rangle +(4\pi \alpha_2 -i\sqrt{-\omega
-i\lambda}) \langle \Gamma_{\lambda^2}, f\rangle]G_{\omega -i\lambda},$$
$$\Sigma_1 f = -[i\lambda (4\pi \alpha_2 -i\sqrt{-\omega +i\lambda})
\langle \mathcal{G}_{\lambda^2}, f \rangle -(4\pi \alpha_1 -i\sqrt{-\omega
+i\lambda}) \langle \Gamma_{\lambda^2}, f \rangle]G_{\omega +i\lambda} +$$
$$+[i\lambda (4\pi \alpha_2 -i\sqrt{-\omega -i\lambda})
\langle \mathcal{G}_{\lambda^2}, f \rangle +(4\pi \alpha_1 -i\sqrt{-\omega
-i\lambda}) \langle \Gamma_{\lambda^2}, f \rangle]G_{\omega -i\lambda},$$
$$\Sigma_2 f = -[i\lambda (4\pi \alpha_1 -i\sqrt{-\omega +i\lambda})
\langle \mathcal{G}_{\lambda^2}, f\rangle -(4\pi \alpha_2 -i\sqrt{-\omega
+i\lambda}) \langle \Gamma_{\lambda^2}, f\rangle ]G_{\omega +i\lambda} +$$
$$+[i\lambda (4\pi \alpha_1 -i\sqrt{-\omega -i\lambda})
\langle \mathcal{G}_{\lambda^2}, f \rangle +(4\pi \alpha_2 -i\sqrt{-\omega
-i\lambda}) \langle \Gamma_{\lambda^2}, f \rangle]G_{\omega -i\lambda}.$$

\noindent The spectrum of the operator $L$ can be decomposed into an
essential and a discrete part,
\begin{equation}    \label{spettro-L}
\sigma(L) = \sigma_{ess}(L) \cup \sigma_d(L),
\end{equation}

\noindent where the essential spectrum is
$$\sigma_{ess}(L) =  \mathcal{C}_+ \cup \mathcal{C}_- = \{ \lambda \in \C: \;
\Re(\lambda) = 0 \, \textrm{and} \, \Im(\lambda) \geq \omega \} \cup
\{ \lambda \in \C: \; \Re(\lambda) = 0 \, \textrm{and} \,
\Im(\lambda) \leq -\omega \},$$

\noindent and the discrete spectrum depends on the paremeter
$\sigma$ as follows:
\begin{itemize}
 \item[(a)] if $\sigma \in (0, 1/\sqrt{2}]$, then the only eigenvalue
of $L$ is $0$
with algebraic multiplicity $2$.
 \item[(b)] if $\sigma \in (1/\sqrt{2}, 1)$, then $L$ has two simple eigenvalues
$\pm i 2\sigma \sqrt{1 -\sigma^2} \omega$ and the eigenvalue $0$
with algebraic multiplicity $2$.
 \item[(c)] if $\sigma = 1$, then the only eigenvalue of $L$ is $0$
with algebraic
multiplicity $4$.
 \item[(d)] if $\sigma \in (1, +\infty)$, then $L$ has two simple
eigenvalues $\pm
2\sigma \sqrt{\sigma^2 -1} \omega$ and the eigenvalue $0$ with
algebraic multiplicity $2$.
\end{itemize}

\end{theo}

Before giving the proof, we need two preliminary lemmas.
\begin{lemma}
 For any $\mu \in \C$, $\omega > 0$, the Green's function ${\mathcal
   G}_\mu$ of the operator
$\mathcal{H_\mu}$, defined by
$$D(\mathcal{H_\mu}) = H^4(\R^3), \quad \mathcal{H}_\mu = \mu +(-\triangle
+\omega)^2, $$
 reads
\begin{equation}    \label{green4}
 \mathcal{G}_{\mu}(x) = \frac{1}{2i \sqrt{\mu}} \left( G_{\omega
-i\sqrt{\mu}}(x) -G_{\omega +i\sqrt{\mu}}(x) \right).
\end{equation}
\end{lemma}

\begin{proof}
By definition of Green's function, $\mathcal{G}_{\mu}$ solves the equation
$[\mu +(-\triangle +\omega)^2] \mathcal{G}_{\mu}(x) = \delta(x).$
Taking the Fourier transform, one
gets
$$\widehat{\mathcal{G}_{\mu}}(k) = \frac{1}{(2\pi)^{3/2} (\mu +(k2
+\omega)2)} 
= \frac{1}{2i \sqrt{\mu}} \left( \widehat{G_{\omega
-i\sqrt{\mu}}}(k) -\widehat{G_{\omega +i\sqrt{\mu}}}(k) \right),$$

\noindent where the function
$G_{\omega \pm i\sqrt{\mu}}$ was defined in \eqref{gom}.
The proof is complete.
\end{proof}

\begin{rmk}
The function $\mathcal{G}_{\mu}$ is an element of $H^s (\R^3)$ for any
$s < 7/2$.
\end{rmk}

\noindent Let us denote
$$\mathcal{H}_\mu^{21} = \mu +L_2 L_1,$$ where $L_2$ and $L_1$ were
defined in \eqref{elleunoedue}.
 Applying elementary rules on composition of operators, one can easily
 see that
the domain of the operator $\mathcal{H}_\mu^{21}$, which coincides with
the domain of
$L_2 L_1$, is given by
\begin{equation} \label{d12} D(L_2 L_1) = \left\{ u \in L^2(\R^3): \;
u = \xi +p G_{\omega
+i\sqrt{\mu}} +q G_{\omega -i\sqrt{\mu}}, \, \textrm{with} \;
\xi \in H^4(\R^3), \; p,q \in \C, \right.
\end{equation}
$$\left. \; \xi(0) +ip \frac{\sqrt{-\omega -i
\sqrt{\mu}}}{4\pi} +iq \frac{\sqrt{-\omega +i \sqrt{\mu}}}{4\pi} =
\alpha_1 (p+q), \right.$$
$$\left. (-\triangle +\omega)\xi(0) +\sqrt{\mu} p \frac{\sqrt{-\omega -i
\sqrt{\mu}}}{4\pi} -\sqrt{\mu} q \frac{\sqrt{-\omega +i
\sqrt{\mu}}}{4\pi} = \alpha_2 i\sqrt{\mu} (q-p) \right\}.$$

\noindent In the following lemma the inverse operator of
$\mathcal{H}_\mu^{21}$ is constructed.
\begin{lemma} \label{045}
For each $\mu \in \C$, the inverse of the operator $\mathcal{H}_\mu^{21}$
is given by
\begin{equation}    \label{ansatz}
(\mathcal{H}_\mu^{21})^{-1}: L^2(\R^3) \rightarrow D(\mathcal{H}_\mu^{21})
\qquad f \mapsto  \mathcal{G}_{\mu}* f \, +p (f) G_{\omega
+i\sqrt{\mu}} +q (f) G_{\omega -i\sqrt{\mu}},
\end{equation}

\noindent where the functionals $p$,$q: L^2(\R^3) \rightarrow \C$ act as
\begin{equation} \label{piqqu} \begin{array}{ll}
    p (f) = \frac{4\pi}{i \sqrt{\mu} W(\mu)} [i \sqrt{\mu} (4\pi \alpha_2 -i
\sqrt{-\omega +i \sqrt{\mu}})\langle \mathcal{G}_{\mu}, f \rangle
    -(4\pi \alpha_1 -i \sqrt{-\omega +i \sqrt{\mu}}) \langle
    \Gamma_{\mu}, f \rangle]\\
    q (f) = \frac{4\pi}{i \sqrt{\mu} W(\mu)} [i \sqrt{\mu} (4\pi \alpha_2 -i
\sqrt{-\omega -i \sqrt{\mu}})\langle \mathcal{G}_{\mu}, f \rangle
    +(4\pi \alpha_1 -i \sqrt{-\omega -i \sqrt{\mu}}) \langle
    \Gamma_{\mu}, f \rangle],\\
  \end{array}
\end{equation}

\noindent with
$$W(\mu) = 2 (4\pi)2 \alpha_1 \alpha_2 -4i \pi (\alpha_1 +\alpha_2)
\left(\sqrt{-\omega +i\sqrt{\mu}} +\sqrt{-\omega -i \sqrt{\mu}}
\right) -2\sqrt{-\omega +i\sqrt{\mu}} \sqrt{-\omega -i\sqrt{\mu}}.$$

\end{lemma}

\begin{proof}
\noindent First we show that the definition of the functionals
$p$ and $q$ ensures
$$ \mathcal{G}_{\mu}*f +p (f) G_{\omega +i\sqrt{\mu}} +q (f)
G_{\omega -i\sqrt{\mu}} \ \in \ D(\mathcal{H}_\mu^{21}) = D(L_2 L_1)$$ for
all $f \in L^2(\R^3)$. Indeed,  $p(f)$ and $q(f)$ solve the
algebraic
system given by the bounday condition in the definition of the domain
\eqref{d12}, namely
$$\left\{
  \begin{array}{ll}
    \langle \mathcal{G}_{\mu}, f \rangle +ip \frac{\sqrt{-\omega -i
\sqrt{\mu}}}{4\pi}
+iq
\frac{\sqrt{-\omega +i \sqrt{\mu}}}{4\pi} = \alpha_1 (p+q)\\
    \langle \Gamma_{\mu}, f \rangle +\sqrt{\mu} p \frac{\sqrt{-\omega -i
\sqrt{\mu}}}{4\pi} -\sqrt{\mu} q \frac{\sqrt{-\omega +i
\sqrt{\mu}}}{4\pi} =
    \alpha_2 i\sqrt{\mu} (q-p).\\
  \end{array}
\right.$$

\noindent
Now, denote by $\widehat H_0$ the operator that acts as the Laplacian on
the subspace of the Schwartz functions in $\R^3$ that vanish in a
neighbourhood of the origin. It is well-known (see \cite{Albeverio}),
that both selfadjoint operators $H_{\alpha_1}$ and $H_{\alpha_2}$ defined in
Proposition \ref{E''} are restrictions of $\widehat H_0^*$ (i.e. the
adjoint of $\widehat H_0$ as an operator in $L^2 (\R^3)$), whose
action on  $ G_{\omega \pm
i\sqrt{\mu}}$ yields
\be \label{sero}
[\mu + (\widehat H_0^* +\omega)^2] G_{\omega \pm i\sqrt{\mu}}
\ = \ 0. \ee

\noindent 
Recalling that
$\mathcal{G}_{\mu} \in H^4(\R^3)$, it follows, for any $f \in L^2 (\R^3)$,
$$\mathcal{H}_\mu^{21}  (\mathcal{G}_{\mu}*f +p (f) G_{\omega
  +i\sqrt{\mu}} +q (f)
G_{\omega -i\sqrt{\mu}} )
\ = \ (\mu +(\widehat H_0^*
+\omega)^2) (\mathcal{G}_{\mu}*f +p (f) G_{\omega +i\sqrt{\mu}} +q (f)
G_{\omega -i\sqrt{\mu}} ) =$$
$$= (\mu +(-\triangle +\omega)^2)(\mu +(-\triangle +\omega)^2)^{-1} f = f.$$

\noindent To conclude the proof one has to show
$$\mathcal{G}_{\mu}*(\mathcal{H}_\mu^{21} f )  +p
(\mathcal{H}_\mu^{21} f) G_{\omega
  +i\sqrt{\mu}} +q (\mathcal{H}_\mu^{21} f)
G_{\omega -i\sqrt{\mu}}   = f$$
for any $f \in
D(\mathcal{H}_{21})$.  To this purpose let us set $f = \xi +a
G_{\omega+i\sqrt{\mu}} +b G_{\omega-i\sqrt{\mu}}$ for some
$\xi \in H^4(\R^3)$ and $a$, $b \in \C$ such that the boundary
condition in \eqref{d12} are satisfied, then, by \eqref{sero}
$$ \mathcal{H}_\mu^{21} f \ = \ [ \mu + ( - \Delta
  + \omega)^2 ] \xi$$
and, by system \eqref{piqqu}
 $$
 p (f) = a, \qquad q(f) = b.
$$
The proof is complete.
\end{proof}

\begin{rmk} \label{046}
{\em The inverse of the operator ${\mathcal H}_\mu^{12} =
    \mu + L_1 L_2$ is obtained exchanging $\alpha_1$ and $\alpha_2$ in the
expression of $({\mathcal H}_\mu^{21})^{-1}$.}
\end{rmk}

Now we can turn to the proof of Theorem \ref{ris}.

\begin{proof}
We preliminarily observe that
$$\Gamma_{\mu}(x)  = (-\triangle +\omega) \mathcal{G}_{\mu}(x) = \frac{e^{i
\sqrt{-\omega +i\sqrt{\mu}} |x|} +e^{i \sqrt{-\omega -i\sqrt{\mu}}
|x|}}{8\pi |x|} = \frac{1}{2} \left( G_{\omega -i\sqrt{\mu}}(x)
+G_{\omega +i\sqrt{\mu}}(x) \right).$$
As proven in Appendix \ref{risolvente}, the following identity holds:
$$R(\lambda) = (L -\lambda I)^{-1} = \left[
  \begin{array}{cc}
    -\lambda (\lambda^2 +L_2 L_1)^{-1} & -L_2 (\lambda^2 +L_1 L_2)^{-1}\\
    L_1 (\lambda^2 +L_2 L_1)^{-1} & -\lambda (\lambda^2 +L_1 L_2 )^{-1}\\
  \end{array}\right]  = \left[
  \begin{array}{cc}
    -\lambda ({\mathcal H}_{\lambda^2}^{21})^{-1} & - L_2  ({\mathcal
      H}_{\lambda^2}^{12})^{-1} \\ L_1  ({\mathcal
      H}_{\lambda^2}^{21})^{-1} & - \lambda  ({\mathcal
      H}_{\lambda^2}^{12})^{-1} \\
 \end{array}
\right],$$

\noindent with $\lambda$ in the resolvent set of $L$, to be specified.

\noindent In order to find the explicit expressions for $\Lambda_1$
and $\Lambda_2$ given in \eqref{lambdasigma}, one sets $\lambda =
\sqrt \mu$ and then applies Lemma \ref{045}, Remark \ref{046}, and
uses the definition of $p$ and $q$ given in
\eqref{piqqu}. Besides, the operators
$\Sigma_1$ and
$\Sigma_2$ can be obtained applying $L_1$ and $L_2$ to ${(\mathcal
      H}_{\lambda2}^{21})^{-1}$
and ${(\mathcal
      H}_{\lambda2}^{12})^{-1}$, respectively, and using some elementary algebra.

\noindent The statement about the essential spectrum of $L$ is a
consequence of Weyl's theorem (Theorem XIII.4 in
\cite{RS4}). On the other hand, the eigenvalues of $L$ are given by
the poles of the resolvent \eqref{eq:risolvente}, or equivalently by
the complex roots of the function $W(\lambda)$; these can be computed
through a lengthy but elementary calculation, here omitted.
\end{proof}

\begin{rmk}{\em
As a by-product, the previous analysis of the complex roots of $W(\lambda)$
reveals the presence of a resonance at the endpoints of
essential spectrum for the case $\sigma=\frac{1}{\sqrt 2}$. }
\end{rmk}

\subsection{Dispersive estimates for the linearized problem in the
case $\sigma \in (0, 1/\sqrt{2})$}\label{dispersive}
In this section we focus on the case $\sigma \in (0, 1/\sqrt{2})$ and
study the behaviour for large $t$ of the propagator
$e^{-Lt}$ restricted to the subspace associated to the essential
spectrum of the operator $L$.

\noindent In order to achieve an effective estimate, the following weighted
$L^p$ spaces are needed
$$L^1_{w}(\R^3) = \left\{ f: \R^3 \rightarrow \C : \; \int_{\R^3}
 w(x) |f(x)| dx < +\infty \right\},$$

\noindent and
$$L^{\infty}_{w^{-1}}(\R^3) = \left\{ f: \R^3 \rightarrow \C : \;
\mathop{\textrm{esssup}}_{x \in \R^3} (w(x))^{-1} |f(x)| < +\infty
\right\},$$

\noindent where $w(x) = 1 +\frac{1}{|x|}$.
The use of such spaces is due to the
singularity of the elements of \eqref{dha}. A
similar choice was made
in \cite{DPT} for the sake of deriving dispersive estimates in the
case of $N$  delta interactions
in $\R^3$.
\begin{theo}    \label{stima-cont}
There exists a constant $C > 0$ such that
$$\left| \frac{1}{2\pi i} \int_{\R^3} \int_{\mathcal{C}_+ \cup \mathcal{C}_-}
(R(\lambda +0) -R(\lambda -0)) (x) e^{-\lambda t} f(y) \, d\lambda  dy
\right| \leq C \left( 1 +\frac{1}{|x|} \right) t^{-\frac{3}{2}}
\int_{\R^3} \left( 1 +\frac{1}{|y|} \right) |f(y)| dy$$

\noindent for any $f \in L^1_{w}(\R^3)$, where
$$
\mathcal{C}_+  = \{ \lambda \in \C: \;
\Re(\lambda) =
0 \, \textrm{and} \, \ \Im(\lambda) \geq \omega \},\ \ \   \mathcal{C}_- =
\{ \lambda \in \C: \; \Re(\lambda) = 0 \, \textrm{and} \, \
\Im(\lambda) \leq -\omega \}\ .
$$
\end{theo}

\begin{proof}
One can compute the propagator $e^{-Lt}$ as the inverse Laplace
transform of the resolvent of $L$. In particular, by Theorem
\ref{ris} and applying the residue theorem, it follows that for $t
> > 0$
$$e^{-Lt} = \frac{1}{2\pi i} \int_{i\R +0} R(\lambda) e^{-\lambda t}
d\lambda =
 \frac{1}{2\pi i} \int_{|\lambda| = r} R(\lambda) e^{-\lambda t}
d\lambda +\frac{1}{2\pi i} \int_{\mathcal{C}_+ \cup \mathcal{C}_-}
(R(\lambda +0) -R(\lambda -0)) e^{-\lambda t} d\lambda,$$

\noindent with $r \in (0, \omega)$ and
$R(\lambda \pm 0) = \lim_{\epsilon \rightarrow 0^+} R(\lambda \pm
\epsilon).$

\noindent We show the computations only for the component
$R_{1,1}(\lambda)$ of the resolvent whose analytic expression is
given in \eqref{eq:risolvente} and \eqref{lambdasigma}, since the
other components can be
handled in the same way.

\noindent Recalling the definition of $\alpha_1$ and $\alpha_2$
given in equation \eqref{alpha}, $R_{1,1}(\lambda)$ can be written
as an integral kernel, namely
\begin{equation}    \label{ris_11}
R_{1,1}(\lambda;x,y) = {i} \frac{e^{i\sqrt{-\omega
+i\lambda}|x-y|} -e^{i\sqrt{-\omega -i\lambda}|x-y|}}{8\pi |x-y|} +
\end{equation}
$$+i \frac{- \sigma \sqrt{\omega} e^{i\sqrt{-\omega -i\lambda}|y|}
e^{i\sqrt{-\omega
+i\lambda}|x|} +[(\sigma +1) \sqrt{\omega} +i \sqrt{-\omega
+i\lambda}] e^{i\sqrt{-\omega -i\lambda}(|x| +|y|)}}{8\pi |x||y| [(2\sigma
+1) \omega +i (\sigma +1)
\sqrt{\omega} \left( \sqrt{-\omega -i\lambda} +\sqrt{-\omega
+i\lambda} \right) -\sqrt{-\omega -i\lambda} \sqrt{-\omega
+i\lambda}]} +$$
$$-i \frac{[(\sigma +1) \sqrt{\omega} +i \sqrt{-\omega
+i\lambda}] e^{i\sqrt{-\omega
+i\lambda}(|x| +|y|)} - \sigma \sqrt{\omega} e^{i\sqrt{-\omega +i\lambda}|y|}
e^{i\sqrt{-\omega -i\lambda}|x|}}{8\pi |x||y| [(2\sigma +1) \omega
+i (\sigma +1) \sqrt{\omega} \left( \sqrt{-\omega -i\lambda}
+\sqrt{-\omega +i\lambda} \right) -\sqrt{-\omega -i\lambda}
\sqrt{-\omega +i\lambda}]}.$$

\noindent Since from equation \eqref{ris_11} it is clear that the
computation of the integral on $\mathcal{C}_+$ and on
$\mathcal{C}_-$ are analogous, we  treat
the cut $\mathcal{C}_+$ only.
On $\mathcal{C}_+$, $\sqrt{-\omega +i\lambda}$ is
continuous while, by the prescription $\Im(\sqrt{-\omega \pm
i\lambda}) > 0$, considering $\epsilon$ as a real parameter,
one has
$$\lim_{\epsilon \rightarrow 0^+} \sqrt{-\omega -i(\lambda +\epsilon)}
= -\lim_{\epsilon \rightarrow 0^+} \sqrt{-\omega -i(\lambda
-\epsilon)} = -\sqrt{-\omega -i\lambda}.$$

\noindent Performing the change of variable $k = \sqrt{-\omega
-i\lambda}$, one can write
$$\int_{\mathcal{C}_+} (R_{1,1}(\lambda +0) -R_{1,1}(\lambda -0)) e^{-\lambda t}
d\lambda = i e^{-i\omega t} \int_{-\infty}^{+\infty} F(k) 2k
e^{-itk2} dk,$$

\noindent where $F$ is the function $R(\lambda +0) -R(\lambda
-0)$ expressed in the variable $k$.

\noindent The function $R_{1,1}$ defined in \eqref{ris_11} is the
sum of a convolution summand $R_{*,1,1}$ and a multiplication
summand $R_{m, 1,1}$, where
$$R_{*, 1,1}(\lambda;x,y) = {i} \frac{e^{i\sqrt{-\omega
+i\lambda}|x-y|} -e^{i\sqrt{-\omega -i\lambda}|x-y|}}{8\pi |x-y|}$$

\noindent and
\begin{equation} \label{rm11}
R_{m, 1,1}(\lambda;x,y) = i \frac{- \sigma \sqrt{\omega}
e^{i\sqrt{-\omega -i\lambda}|y|} e^{i\sqrt{-\omega +i\lambda}|x|}
+[(\sigma +1) \sqrt{\omega} + i \sqrt{-\omega +i\lambda}]
e^{i\sqrt{-\omega -i\lambda}(|x| +|y|)}}{8\pi |x||y| [(2\sigma +1)
\omega +i (\sigma +1) \sqrt{\omega} \left( \sqrt{-\omega -i\lambda}
+\sqrt{-\omega +i\lambda} \right) -\sqrt{-\omega -i\lambda}
\sqrt{-\omega +i\lambda}]} + \end{equation}
$$-i \frac{[(\sigma +1) \sqrt{\omega} +i \sqrt{-\omega
+i\lambda}] e^{i\sqrt{-\omega +i\lambda}(|x| +|y|)} - \sigma
\sqrt{\omega} e^{i\sqrt{-\omega +i\lambda}|y|} e^{i\sqrt{-\omega
-i\lambda}|x|}}{8\pi |x||y| [(2\sigma +1) \omega +i (\sigma +1)
\sqrt{\omega} \left( \sqrt{-\omega -i\lambda} +\sqrt{-\omega
+i\lambda} \right) -\sqrt{-\omega -i\lambda} \sqrt{-\omega
+i\lambda}]}.$$

\noindent Then we can define
$$F_*(k) = R_{*, 1,1}(\lambda +0) -R_{*, 1,1}(\lambda -0), \quad
{\rm{and}} \quad
F_m(k) = R_{m, 1,1}(\lambda +0) -R_{m, 1,1}(\lambda -0).$$

\noindent One can easily compute $F_*$ and gets
$F_*(k) = - \frac{\sin(|x-y| k)}{4\pi |x-y|}.$
Thus, 
by formula 3.851 in \cite{tavole},
$$\int_{-\infty}^{+\infty} F_*(k) 2k e^{-itk2} dk =
\sin(|x-y| k) dk -i
 \frac{1 +i}{16 \sqrt{\pi}} t^{-\frac{3}{2}} e^{i\frac{|x-y|^2}{4t}},$$

\noindent for any $t > 0$. Hence
\begin{equation}    \label{0-conv}
\left| \frac{1}{2 \pi i} \int_{\R^3} \int_{-\infty}^{+\infty} F_*(k;
y) dk f(y) dy \right| \leq \frac{1}{8 \sqrt{2\pi}} t^{-\frac{3}{2}}
\int_{\R^3} |f(y)| dy.
\end{equation}

\noindent Let us estimate $\int_{-\infty}^{+\infty} F_m (k) 2k
e^{-itk^2} \, dk$. One can
notice that $F_m(k)$ is the sum of terms of the form
$\frac{i}{8\pi |x||y|} g(k) e^{\pm iks},$
where $g(k)$ is a rational function of $k$
and $\sqrt{-2 \omega - k^2}$ possibly multiplied by
$e^{i\sqrt{-2\omega -k^2} s}$, and $s$ can be $0$, $|x|$, $|y|$ or
$|x| +|y|$.
Let us consider the term
$$g (k) e^{-ik (|x| +|y|)} =$$
$$= \frac{(\sigma +1) \sqrt{\omega}
+i\sqrt{-2\omega -k^2}}{(2\sigma +1) \omega +i(\sigma +1)
\sqrt{\omega} (\sqrt{-2\omega -k^2} -k) +k \sqrt{-2\omega -k^2}}
e^{-ik (|x| +|y|)},$$
which results from the second term in \eqref{rm11} referred to
$R_{m,1,1} (\lambda + 0)$.

\noindent Notice that $g \in C^1(\R, \C)$ and
$|g (k)| \sim \frac{1}{ik} \quad \textrm{as} \; k \rightarrow
+\infty,$
hence $g \in L^2(\R)$. Moreover,
$$\frac{dg}{dk}(k)
= \frac{-ik}{\left[(2\sigma +1) \omega +i(\sigma +1)
\sqrt{\omega} (\sqrt{-2\omega -k^2} -k) +k \sqrt{-2\omega
-k^2}\right]\sqrt{-2\omega -k^2}} +$$
$$-\frac{(\sigma +1) \sqrt{\omega} +i \sqrt{-2\omega -k^2}}{[(2\sigma +1) \omega
+i(\sigma +1)
\sqrt{\omega} (\sqrt{-2\omega -k^2} -k) +k \sqrt{-2\omega -k^2}]2}
\cdot$$
$$\cdot \left( -\frac{i(\sigma +1) \sqrt{\omega}
k}{\sqrt{-2\omega -k^2}} -i(\sigma +1) \sqrt{\omega} +\sqrt{-2\omega
-k2} -\frac{k^2}{\sqrt{-2\omega -k^2}} \right),$$

\noindent which belongs to $L^2(\R)$ too, so
$g$ is an element of $H^1(\R)
$, and as consequence $\check{g}  \in L^1(\R
)$, where $\check{g}$ is the inverse
Fourier transform of $g$.
Furthermore, one can compute the inverse Fourier transform
of $2k e^{-it k^2}$ as
$$U_t(s) = \frac{1}{2\pi i} \int_{-\infty}^{+\infty} 2k e^{-it k^2}
e^{-ik s} dk = \frac{1}{(4\pi i t)^\frac{3}{2}} e^{-\frac{s^2}{4i
t}}.$$
From the last identity it follows
$$\left| \frac{1}{2\pi i} \int_{\R^3} \int_{-\infty}^{+\infty}
\frac{i}{8\pi |x||y|} g (k) e^{-ik (|x| +|y|)} 2k e^{-it k2} dk
f(y)
dy \right| =$$
\begin{equation}    \label{0-molt}
= \left| \int_{\R^3} \int_{-\infty}^{+\infty} \frac{1}{8\pi
|x||y|} \check{g}(u) U_t(u-|x| -|y|) du f(y) dy
\right|
\leq C \frac{1}{|x|} t^{-\frac{3}{2}} \int_{\R^3} \frac{|f(y)|}{|y|} dy,
\end{equation}

\noindent where the last inequality follows from H\"{o}lder
inequality and $C > 0$.
The other terms in $F_m(k)$ are handled in an analogous way so we do
not give details.

Summing up, let $f \in L^1_w(\R^3)$. Then
$$\left| \frac{1}{2\pi i} \int_{\R^3} \int_{\mathcal{C}_+ \cup \mathcal{C}_-}
(R(\lambda +0) -R(\lambda -0)) e^{-\lambda t} d\lambda f(y) dy
\right| \leq$$
$$\leq \frac{1}{2\pi} \left( \int_{\R^3} \left| \int_{\mathcal{C}_+}
(R(\lambda +0) -R(\lambda -0)) e^{-\lambda t} d\lambda f(y) \right|
dy +\right.$$
$$\left.+\int_{\R^3} \left| \int_{\mathcal{C}_-} (R(\lambda +0) -R(\lambda
-0)) e^{-\lambda t} d\lambda f(y) \right| dy \right) =
\frac{1}{2\pi} (I +II).$$

\noindent Let us estimate the integral $I$. Thanks to the estimates
\eqref{0-conv} and \eqref{0-molt} one has
$$I  = \int_{\R^3} \left| \int_{-\infty}^{+\infty} F(k) 2k e^{-itk^2}
dk f(y) \right| dy \leq$$
$$\leq \int_{\R^3} \left| \int_{-\infty}^{+\infty} F_*(k) 2k
e^{-itk^2} dk f(y) \right| dy +\int_{\R^3} f(y) \left|
\int_{-\infty}^{+\infty} F_m(k) 2k e^{-itk^2} dk  \right| dy
\leq$$
$$\leq C t^{-3/2} \left( \int_{\R^3} |f(y)| dy +\frac{1}{|x|}
\int_{\R^3} \frac{|f(y)|}{|y|} dy \right) \leq
C \left( 1 +\frac{1}{|x|} \right) t^{-3/2} \int_{\R^3} |f(y)|
\left( 1 +\frac{1}{|y|} \right) dy.$$

\noindent The integral $II$ can be estimated in the same way, which
completes the proof.
\end{proof}

\begin{rmk}	\label{split}{\em
Evaluating the propagator $e^{-Lt}$ at $t = 0$ one gets
$$1 = \frac{1}{2\pi i} \int_{|\lambda| = r} R(\lambda) d\lambda
+\frac{1}{2\pi i} \int_{\mathcal{C}_+ \cup \mathcal{C}_-} (R(\lambda
+0) -R(\lambda -0)) d\lambda = P_0 +P_c.$$

\noindent From Lemma \ref{proiettore} it will follow that the operators
$P_0$ and $P_c$ are symplectic projectors onto the subspaces
associated to generalized
kernel and to the continuous spectrum respectively.
Finally, let us note that explicitly integrating the resolvent around
its poles it
turns out that the dynamics along the generalized kernel grows
linearly in time. This fact is
proved in Appendix \ref{polo}.}
\end{rmk}

\section{Modulation equations} \label{modula}
In this section we restrict to the case $\sigma \in (0,
1/\sqrt{2})$, summarize the main technical steps
and  give some preliminary results towards the proof of asymptotic
stability of standing waves. In particular, we write the
so-called
{\em modulation equations} that rule the evolution of a perturbed
standing wave when splitted in a solitary component and a fluctuating one.
We recall once more that the scalar product we adopt is the real
scalar product on the Hilbert space $L^2(\R^3,\C)$ defined
in \eqref{scalarproduct}. 
\noindent
In order to make the reading easier, let
us give a
brief outline of the strategy to be employed. We follow the roadmap of the classical
papers \cite{SW1},\cite{SW2},\cite{BP1},\cite{BP2}, \cite{BS}, also
adopted for the model with concentrated
nonlinearity in dimension one in \cite{BKKS} and \cite{KKS}. More
specifically, we decompose
the dynamics  in the neighbourhood
of the solitary manifold in  a ``longitudinal" and a ``transversal"
component with respect to the generalized kernel $N_g (L)$, given in
Theorem \ref{kergen}, of the linearized operator
$L$. In order to perform the required analysis, we
exploit the symplectic structure introduced in
Section \ref{sw}.
Let us begin by noticing that the solitary manifold $\mathcal{M}$
defined 
in \eqref{M} is a symplectic
submanifold of $(L^2(\R^3,\C), \Omega)$, invariant under the flow of \eqref{eq1}. Its
tangent space at the standing wave $e^{i\theta}\Phi_\omega$ when $\theta=0$ is two-dimensional
and is generated by the vectors
$\frac{d}{d\theta}\{e^{i\theta}\Phi_{\omega} \}_{\theta=0}$ and
$\frac{d}{d\omega}\{e^{i\theta}\Phi_{\omega} \}_{\theta=0}$, in real
representation given by
$$
\frac{d}{d\theta}\{e^{i\theta}\Phi_{\omega} \}_{\theta=0}\ \mapsto \ e_1 =\left(
  \begin{array}{ll}
   0\\
   \Phi_{\omega}
   \end{array}
   \right)\ \ {\rm and} \ \
   \frac{d}{d\omega}\{e^{i\theta}\Phi_{\omega} \}_{\theta=0}\ \mapsto \ e_2= \left(
   \begin{array}{ll}
   \varphi_{\omega}\\
   0
   \end{array}
\right)\ ,
$$

\noindent  where $\varphi_{\omega}=\frac{d}{d\omega}\Phi_{\omega}$ was defined in
   Section \ref{linearizzato}. However, when no confusion arises we
   use the shorthand expressions $\Phi_{\omega}$ and $\varphi_\omega$
   with the meaning of the corresponding {\it real} representative
   vectors (second component vanishing). As already remarked the
   couple of vectors $\{e_1, e_2\}$ is a basis for $N_g (L)$. It
   immediately seen that
   $\Omega(e_1,e_2)=\frac 1 2 \frac{d}{d\omega}\|\Psi_\omega\|^2\neq 0$, thanks
   to the condition $\sigma \in (0,1/\sqrt{2})$ guaranteeing orbital
   stability. So the symplectic form is nondegenerate on the solitary
   manifold $\mathcal{M}$, which is a symplectic submanifold. By its
   very definition, $\mathcal{M}$ is invariant for the flow
   of \eqref{eq1}.  \par\noindent The following lemma
   establishes the relation between the spectral projection $P_0$
   introduced in Remark \ref{split} and the symplectic projection onto
   the solitary manifold.
\begin{lemma}   \label{proiettore}
Let $\Delta = \frac 1 2 \frac{d}{d\omega} \| \Phi_{\omega} \|_{L^2}^2,$ then for any $f \in L^2(\R^3)$
\be\label{Pzerosymp}
P_0 f = \frac{1}{\Delta} \Omega \left( f,
\varphi_\omega \right) J\Phi_{\omega}
-\frac{1}{\Delta} \Omega \left( f, J\Phi_{\omega} \right)
\varphi_\omega\ ,
\ee

\noindent where $\Omega(\cdot, \cdot)$ was defined in \eqref{formasimplettica}.

\end{lemma}

\begin{proof}
The explicit expression of
the spectral projection $P_0 = \frac{1}{2\pi i} \int_{|\lambda| = r} R(\lambda)
d\lambda$ can be recovered by Appendix \ref{polo}, and the equivalence
with the r.h.s. follows by straightfoward  calculations.
\end{proof}  
Notice that the given representation of $P_0$ is well defined thanks to the fact that $\Delta>0$, again as a consequence of the choice $\sigma \in (0,
1/\sqrt{2})$. Moreover, $P_0$ is a symplectically orthogonal
projection, in the sense that given a couple $\{\zeta,f\}$ with
$\zeta\in {\rm Im}\ P_0$ and $f\in {\rm Ker}\ P_0$, one has
$\Omega(\zeta, f)=0\ .$ 
In particular it is useful to note that due to the definition of
symplectic form $\Omega$, a state $f$ with vanishing component along
the continuous spectrum of $L$ is orthogonal to the vectors $Je_1$ and
$Je_2$, or in complex notation, to $\Phi_\omega$ and
$i\frac{d}{d\omega}\Phi_\omega\ =i\varphi_{\omega}.$  \par\noindent 
After these preliminaries, as anticipated in formula \eqref{asrepr},
we write the solution to \eqref{eq1} as
\begin{equation}    \label{mod-ansatz}
u(t,x) = e^{i\Theta(t)} \left( \Phi_{\omega(t)}(x) +\chi(t,x)
\right),\ \  \Theta(t) = \int_0^t \omega(s) ds +\gamma(t) ,
\end{equation}
with the final goal of proving that the solution decomposes in the sum of a solitary component and a dispersive one.
\par\noindent
  
The local splitting of the invariant symplectic
   manifold $(L^2(\R^3,\C), \Omega)$ in two symplectically orthogonal
   manifolds, the finite dimensional solitary manifold $\mathcal{M}$
   and the infinite dimensional range of the spectral projection on
   the continuous spectrum, suggests to symplectically project the flow
   according to this decomposition (see also Remark \ref{split}),
   in order to obtain the so called modulation equations. The
   projection along $\mathcal{M}$ (``longitudinal") gives rise to two
   ordinary 
   differential equations for the frequency $\omega$ and the phase
   $\gamma$ of the solitary wave, depending parametrically on the fluctuating component $\chi$; while the projection on the
   continuous spectrum (``transversal") gives a partial differential
    equation for the remainder $\chi$ (with coefficients depending on $\gamma$ and $\omega$).
The solution to the equation for the $\chi$ component will be
shown to decay in time in suitable norms. As a consequence, one has the
asymptotic behaviour of the solutions for the parameters $\omega$ and
$\gamma$ of the 
solitary wave, to be shown in Section 6, and finally asymptotic
stability, which will be the subject of Section 7. 
\par\noindent
To deduce the modulation equations
it proves convenient to make use of the  variational formulation of equation \eqref{eq1}
\begin{equation}    \label{eq-variazionale}
\left( i\frac{du}{dt}(t), v \right)_{L^2} = E'[u(t)](v) \quad
\forall v \in V.
\end{equation}

\noindent 
To begin with, we replace in the previous equation the Ansatz \eqref{mod-ansatz}.

\noindent \noindent 

\noindent By equation \eqref{gs_eq} and Proposition \ref{E''}, equation \eqref{eq-variazionale} can be rephrased as
\begin{equation}    \label{eqchi}
\left( i\frac{d\chi}{dt}(t), v \right)_{L^2} = Q_{\alpha,
Lin}(\chi(t),v) +\dot{\gamma}(t) (\Phi_{\omega(t)} +\chi(t),
v)_{L^2} +\dot{\omega}(t) \left(-i\frac{d\Phi_{\omega(t)}}{d\omega}
,v\right)_{L^2} +N(q_{\chi}(t),q_v(t))
\end{equation}

\noindent for any $v \in V$. \par\noindent Here $Q_{\alpha,Lin}$ is the quadratic
form of  the operator $D$ defined in \eqref{eq_linearizzata0} and acting as
$$Q_{\alpha,Lin}(\chi,v) = (\nabla \phi_{\chi}, \nabla \phi_v)_{L^2}
-\frac{\sqrt{\omega}}{4\pi} \Re (q_{\chi} \overline{q_v})
-\sigma \frac{\sqrt{\omega}}{2\pi} \Re q_{\chi} \Re q_v +\omega (\chi,v)_{L^2}\ , $$

\noindent and the nonlinear remainder $N(q_{\chi},q_v)$ is given by
$$N(q_{\chi}, q_v) = -\nu |q_{\chi}
+q_{\omega}|^{2\sigma} \Re ((q_{\chi} +q_{\omega}) \overline{q_v})
+\nu (2\sigma +1)|q_{\omega}|^{2\sigma} \Re q_{\chi} \Re q_v +\nu
|q_{\omega}|^{2\sigma} \Im q_{\chi} \Im q_v +\nu
|q_{\omega}|^{2\sigma} \Re(q_{\omega} \overline{q_v}).$$ 
In the previous equation, according to Section II B, $q_\omega =
\left( \frac {\sqrt \omega} {4 \pi \nu} \right)^{\frac 1 {2 \sigma}}.$

\par\noindent
\begin{rmk}{\em
The remainder $N(q_{\chi}, q_v)$ depends nonlinearly on $\chi$ (and
$\omega$) and it is real linear in $v$; so, by Riesz representation
theorem and with a slight abuse of notation, there exist a vector
$N(q_\chi)$ such that $N(q_{\chi},q_v)=\Re
N(q_\chi)\overline{q_v}$. It is a peculiarity of this model that in
fact it depends just on the charges of $\chi$ and $v$. Moreover, by
its very definition, the remainder is the difference between the
action of the complete vector field and its linear part at the
solitary wave, and so it is quadratic in $q_{\chi} $ near $\chi=0$.}  
\end{rmk}
\par\noindent
Corresponding expressions can be given with obvious modification in purely real form, which we omit for the sake of brevity.
\noindent Since $\omega$, $\gamma$ and $\chi$ are all unknown the Ansatz \eqref{mod-ansatz} makes the problem underdetermined, and a supplementary condition is needed to give a unique representation of the solution; a way to close the 
system for $\omega$, $\gamma$ and $\chi$ is to require that the $\chi$
component is decoupled from the discrete spectrum, i.e. $P_0\chi=0$,
or equivalently to project equation \eqref{eqchi} onto the
symplectically orthogonal complement of the generalized kernel of
$L$. The corresponding modulation equations take different forms
according to the way one writes the projection and we give two of them
for future reference. In the following we denote by $Q_L$ the bilinear
form associated to the linear nonselfadjoint operator $L$.

\noindent 

\begin{theo}{\emph{\textbf{(Modulation equations I)}}}  \label{lemma-modI}
Let $\chi$ be a solution to equation \eqref{eqchi} such that
$P_0 \chi (t) = 0$ for all $t \geq 0$, and let the functions $\omega$ and $\gamma$
belong to ${\rm C}^1(\R)$; then $\omega$ and $\gamma$
solve the equations 
\begin{equation}    \label{mod-eq1}
\dot{\omega} = \frac{\Re\ (JN(q_{\chi}) \overline{q_{P_0^*(\Phi_\omega+\chi)}})}{\left(
\varphi_{\omega} - \frac{d P_0}{d\omega}\chi , \Phi_\omega + \chi \right)_{L^2}},
\end{equation}

\noindent and
\begin{equation}    \label{mod-eq2}
\dot{\gamma} = \frac{\Re\ (JN(q_{\chi}) \overline{q_{J({\varphi_\omega}-\frac{d P_0}{d\omega}\chi)}}} {\left(
\varphi_{\omega} - \frac{d P_0}{d\omega}\chi , \Phi_\omega + \chi \right)_{L^2}}
\end{equation}

\end{theo}

\begin{proof}
We adapt the reasoning in \cite{BS}. Equation \eqref{eqchi} is equivalent to
\be\label{eqchi2}
\left(\frac{d}{dt}(\Phi_\omega +\chi), v\right)_{L^2} = Q_L(\chi, v)) +\dot\gamma\left(J(\Phi_\omega +\chi),v\right)_{L^2} + \Re(JN(q_\chi)\overline{q_v})\ \ \ \forall v\in V\ .
\ee
Set $v=P_0^*(\Phi_\omega +\chi)$; notice that differentiating in time
$P_0\chi=0$, one has $$P_0\frac{d}{dt}(\Phi_\omega
+\chi)=\dot\omega\left(\varphi-\frac{dP_0}{d\omega}\chi \right)\ ,$$
where expressions such as $\frac{dP_0}{d\omega}\chi$ are computed
from the representation given in \eqref{Pzerosymp}. \par\noindent 
Moreover, one immediately has the identities
$$
Q_L(\chi,P_0^*(\Phi_\omega + \chi))=Q_L(P_0\chi,(\Phi_\omega + \chi))=0
$$
and, using $P_0J=JP_0^*$,
$$
\left(J(\Phi_\omega + \chi),P_0^*(\Phi_\omega + \chi)\right)_{L^2}=\left(J(\Phi_\omega + \chi),(P_0^*)^2(\Phi_\omega + \chi)\right)_{L^2}=\left(JP_0^*(\Phi_\omega + \chi),P_0^*(\Phi_\omega + \chi)\right)_{L^2}=0\ .
$$
So one remains with
$$
\left(P_0\frac{d}{dt}(\Phi_\omega +\chi), \Phi_\omega + \chi\right)_{L^2}=\Re(JN(q_\chi)\overline{q_{P_0^*(\Phi_\omega + \chi)}})
$$
from which the equation for $\dot\omega$ follows.\par\noindent
Now let us consider the test function $JP_0\frac{d}{dt}(\Phi_\omega +\chi)$, and notice the following facts, in which use is made of $JP_0=P_0^*J$.\par\noindent
$$
\left(\frac{d}{dt}(\Phi_\omega +\chi), JP_0(\Phi_\omega + \chi)\right)_{L^2}=\left(\frac{d}{dt}(\Phi_\omega +\chi), JP_0^2(\Phi_\omega + \chi)\right)_{L^2}=\left(P_0\frac{d}{dt}(\Phi_\omega +\chi), JP_0(\Phi_\omega + \chi)\right)_{L^2}=0\ ;
$$
$$
Q_L(\chi,JP_0\frac{d}{dt}(\Phi_\omega +\chi))=0\ .
$$
It follows from the weak equation \eqref{eqchi2} 
$$
\dot\gamma\left(\Phi_\omega +\chi,P_0\frac{d}{dt}(\Phi_\omega +\chi)\right)_{L^2}=\Re(JN(q_\chi)\overline{q_{JP_0\frac{d}{dt}(\Phi_\omega + \chi)}}
$$
and hence, after substituting the expression of
$P_0\frac{d}{dt}(\Phi_\omega +\chi)$ determined above and cancelation
of $\dot\omega$ the equation for $\dot\gamma$ follows. This ends the
proof. 
\end{proof}

\par\noindent
Two properties of the modulation equations which will be useful in the
subsequent analysis are the following.\par\noindent 
\begin{cor} \label{stime-modulazione}
Under the hypotheses of Theorem \ref{lemma-modI}, and if it is known that $\| \chi
\|_{L^1_w}$ is sufficiently small, the right hand sides of
\eqref{mod-eq1} and \eqref{mod-eq2} are smooth and there exists a
continuous function $\mathcal{R} = \mathcal{R} (\omega, \| \chi
\|_{L^1_w})$ such that, for any $t \geq 0$,
$$|\dot{\omega}(t)| \leq \mathcal{R} |q_{\chi}(t)|^2 \quad \mbox{and}
\quad |\dot{\gamma}(t)| \leq \mathcal{R} |q_{\chi}(t)|^2.$$
\end{cor}

\noindent The proof of the previous result is a consequence of two
facts. In the first place  $(\varphi_\omega,\Psi_\omega)_{L^2}=\frac
1 2 \frac{d}{dt}\|\Phi_\omega\|^2>0 $ by condition $\sigma \in (0,
1/\sqrt{2})$ which gives orbital stability; secondarily, the nonlinear part in \eqref{eqchi} actually depends only on
the charges $q_{\chi}$ and $q_v$; provided that $|q_{\chi}| \leq
c$,  there exists a positive constant $C > 0$ such that the denominators in \eqref{mod-eq1} and \eqref{mod-eq2} are strictly away from zero and
$$|N(q_{\chi})| \leq C |q_{\chi}|^2, \quad \forall \chi \in V.$$

\par\noindent  The second property concerns the compatibility of the
orthogonality condition  of the fluctuating part $\chi$ with arbitrary choices 
of initial data. The following lemma assures in fact that the
orthogonality condition $P_0 \chi = 0$ can be satisfied at the initial
time in the neighbourhood of the solitary manifold without loss of
generality. 
\begin{lemma}   \label{dati-iniz}
Let $u(t) \in C(\R^+, V)$ be a solution to equation \eqref{eq1} with
$u(0) = u_0 \in V \cap L^1_w$ and assume
$$d = \| u_0 -e^{i\theta_0} \Phi_{\omega_0} \|_{V \cap L^1_w} \ll 1,$$

\noindent for some $\omega_0 > 0$ and $\theta_0 \in \R$.

\noindent Then, there exists a stationary wave
$e^{i\widetilde{\theta}_0} \Phi_{\widetilde{\omega}_0}$, and
$\chi_0(x)$ with $P_0(\widetilde{\omega}_0) \chi_0 = 0$ such that
$u_0(x) =  e^{i\widetilde{\theta}_0} \left( \Phi_{\widetilde{\omega}_0}(x) +\chi_0(x) \right),$
and $\| \chi_0 \|_{V \cap L^1_w} = O(d) \quad \textrm{as} \; d \rightarrow 0.$

\end{lemma}

\noindent The result is commonly stated as a preliminary step in the
analysis of modulation equations (see for example \cite{GNT},\cite{KM}
and \cite{BKKS}). The proof is an application of the implicit function
theorem making use again of the condition
$ \frac{d}{dt}\|\Phi_\omega\|^2 \neq 0 $; we omit details and refer to
the quoted references. As a consequence of the previous lemma, in all
proofs in the rest of the paper 
we can assume $P_0\chi_0 = 0$ where $\chi_0=\chi(0)$.\par\noindent

\par\noindent
An equivalent form of the modulation equations for the soliton
parameters $\omega$ and $\gamma$ can be obtained exploiting the
characterization of the condition $P_0\chi=0$ through the (Hilbert)
orthogonality $(\chi, \Phi_{\omega})_{L^2}=0=(\chi,
i\varphi_{\omega})_{L^2}$. In some respects they are more transparent
and we give them making use of the complex writing. 
\begin{theo}{\emph{\textbf{(Modulation equations II)}}}  \label{lemma-modII}
Let $\chi$ be a solution to equation \eqref{eqchi} such that
$P_0 \chi (t) = 0$ for all $t \geq 0$, and let the functions $\omega$
and $\gamma$ 
belong to ${\rm C}^1(\R)$; then $\omega$ and $\gamma$
satisfy the equations 

\be
\dot\omega=\frac{((\chi,\varphi_\omega)_{L^2}+(\varphi_\omega,\Phi_\omega)_{L^2})
N(\chi,
i\Phi_\omega)-(\chi,i\Phi_\omega)_{L^2}N(\chi,\varphi_\omega)}{(\varphi_\omega,\Phi_\omega)_{L^2}^2-(\chi,\varphi_\omega)_{L^2}^2}  
\ee
\\
\be
\dot\gamma=\frac{((\chi,\varphi_\omega)_{L^2}-(\varphi_\omega,\Phi_\omega)_{L^2}) N(\chi, \varphi_\omega)+(\chi,i\frac{d}{d\omega}\varphi_\omega)_{L^2}N(\chi, i\Phi_\omega)}{(\varphi_\omega,\Phi_\omega)_{L^2}^2-(\chi,\varphi_\omega)_{L^2}^2}
\ee

\end{theo}

\begin{proof}
Differentiating in time the orthogonality conditions
$(\chi, \Phi_\omega)_{L^2}=0=(\chi, i\varphi_\omega)_{L^2}$, it easily follows
that 
\begin{equation*}
(i\dot\chi,
i\Phi_\omega)_{L^2}=-\dot\omega(\chi, \varphi_\omega)_{L^2}, \qquad
(i\dot\chi, \varphi_\omega)_{L^2}=\dot\omega(\chi,
i\frac{d}{d\omega}\varphi_\omega)_{L^2}\ . 
\end{equation*}
So testing the weak equation for $\chi$ with $i\Phi_\omega$ and
$\varphi$ and taking into account properties of operators $L_1$ and
$L_2$ and orthogonality conditions again, one obtains the system 
\begin{align*}
\dot\omega((\chi,\varphi_\omega)_{L^2}-
(\Phi_\omega,\varphi_\omega)_{L^2})
+ \dot\gamma(\chi,i\Phi_\omega)_{L^2}&=-N(\chi,i\Phi_\omega)\\ 
\dot\omega(\chi,\frac{d}{d\omega}\varphi_\omega)_{L^2}
- \dot\gamma((\Phi_\omega,\varphi_\omega)_{L^2}+(\chi,\varphi_\omega)_{L^2})&=N(\chi,\varphi_\omega). 
\end{align*}
The thesis follows solving for $\dot\omega$ and $\dot\gamma$.

\end{proof}
\par\noindent
Notice that to this second form of modulation equations apply similar
remarks to the ones made for the first form. In particular, if a
priori estimates on smallness of $\chi$ are known, the modulation
equations are well defined thanks to the condition
$\frac{d}{d\omega}\|\Phi_\omega\|^2>0\ ,$ and the analogous of
Lemma \ref{dati-iniz} holds true.

\section{Time decay of weak solutions}
The goal of this section is to provide the time decay of the
transversal component $\chi$ of the solution $u$ (see (31)) to
equation \eqref{eq1}; 
the result we achieve shows that $\chi$ is in fact not only a fluctuation, but also a
decaying disperive remainder and it paves the way to the proof of
asymptotic stability of standing waves, that is given in the next
section. To this end we follow the idea 
developed in \cite{BP1},\cite{BP2},\cite{BS} for the standard NLS and
applied in \cite{BKKS} to the case of 1-d concentrated
nonlinearities. 
\par\noindent For any $T > 0$, define preliminarily the
so-called  
{\it majorant}
\begin{equation}    \label{maggioranti}
M(T) = \sup_{0 \leq t \leq T} \left[ (1+t)^{3/2} \| \chi(t)
\|_{L^{\infty}_{w^{-1}}} +(1+t)^3 (|\dot{\gamma}(t)|
+|\dot{\omega}(t)|) \right].
\end{equation}

\noindent 
We aim at proving that the majorant is uniformly bounded  in $T$  by a
constant $\overline{M}=O(d)$, where $d$ is the size of the dispersive
component $\chi$.  The proof of such bound is the content of the
following theorem. 

\begin{theo}    \label{teo-maggioranti}
Let $u \in C(\R^+, V)$ be a solution to equation
\eqref{eq1} with $u(0) = u_0 \in V \cap L^1_w$ and
define
$d := \| u_0 -e^{i\theta_0} \Phi_{\omega_0} \|_{V \cap L^1_w},$
for some $\omega_0 > 0$ and $\theta_0 \in \R$. Then, if $d$ is
sufficiently small, there are $\omega, \gamma \in C^1(\R^+)$
which satisfy \eqref{mod-eq1}-\eqref{mod-eq2}, and such that the
solution $u$ can be written as in \eqref{mod-ansatz}. 

\noindent
Moreover,
there is a positive constant $\overline{M} > 0$, depending only on
the initial data, such that, for any $T>0$, one has $M(T) \leq \overline{M},$ 
and $\overline{M} = O(d)$ as $d
\rightarrow 0$.
In particular
\begin{align}
\|\chi(t)\|_{L^{\infty}_{w^{-1}}}&\leq \overline{M}\ (1+t)^{-3/2} \ &\forall t>0\\
|\dot{\gamma}(t)|+|\dot{\omega}(t)|&\leq \overline{M}\ (1+t)^{-3} \
 &\forall t>0. 
\end{align} 
\end{theo}


\noindent The previous theorem is implied by the
following proposition that is proven in
Section \ref{proof-maggioranti} by using the results given in
Sections \ref{frozen} and \ref{duhamel}, and the dispersive properties of the linearization operator $L$ given in Section \ref{dispersive}. 

\begin{prop}    \label{prop-maggioranti}
Under the hypotheses of the previous theorem, assume that there
exist some $t_1 > 0$ and $\rho > 0$ such that $M(t_1) \leq \rho$.
Then there are two positive numbers $d_1$ and $\rho_1$, independent
of $t_1$, such that if $d = \| \chi_0 \|_{V \cap L^1_w} < d_1$ and
$\rho < \rho_1$, then $M(t_1) \leq \frac{\rho}{2}.$

\end{prop}

\noindent Indeed, if Proposition \ref{prop-maggioranti} were true,
then Theorem \ref{teo-maggioranti} would follow from the next
argument:
let $\mathcal{I} \subset [0, +\infty)$ be defined as
$$\mathcal{I} = \{ t_1 \geq 0: \;\; \omega, \gamma \in C^1([0,
t_1]), \; M(t_1) \leq \rho \}.$$

\noindent $\mathcal{I}$ is obviously relatively closed in $[0, +\infty)$ with the topology induced by considering it as a subspace of $\R$ with the
standard Euclidean topology. On the other
hand, the thesis of Proposition \ref{prop-maggioranti} and the
estimates of Corollary \ref{stime-modulazione} imply that
$\mathcal{I}$ is also relatively open. Hence, the uniform estimate
of Theorem \ref{teo-maggioranti} follows from the fact that $\sup
\mathcal{I} = +\infty$.

\subsection{Frozen linearized problem}	\label{frozen}
Note that the equation \eqref{eqchi} is non autonomous. In order to
make its study simpler, it is useful to exploit a further
reparametrization of the solution $\chi(t)$. We fix a time $t_1> 0$
and denote $\omega_1 = \omega(t_1)$ and $\gamma_1
= \gamma(t_1)$. Now define (in vector notation; we recall that $J$
corresponds to $-i$)  
\begin{equation}\label{reparametrization} 
e^{-J\Theta(t)}\chi(t,x) = e^{-J\tilde{\Theta}(t)} \eta(t,x)\ ,\ \ {\rm where}\ \ \tilde{\Theta}(t) = \omega_1 t +\gamma_1\ .
\end{equation} 
The function $\eta$ satisfies the equation

\begin{align*}
\left(e^{J(\Theta -\tilde{\Theta})} \frac{d\eta}{dt}, v \right)_{L^2} &= Q_L(e^{J(\Theta -\tilde{\Theta})} \eta, v)+(\omega_1 -\omega) (J\eta, v)_{L^2}+\dot{\gamma} (J\Phi_{\omega}, v)_{L^2} \\
& -\dot{\omega} \left(
\frac{d\Phi_{\omega}}{d\omega}, v \right)_{L^2} +JN(e^{J(\Theta
-\tilde{\Theta})} q_{\eta}) \overline{q_v}\ \ \ \forall v \in V
\end{align*}

\n
We need a further manipulation which allows to rewrite the previous
equation in a form which makes the role of
reparametrization clear. To this end we need the following identities, which
can be obtained from straightforward computations  

\begin{itemize}
\item $J e^{J(\Theta -\tilde{\Theta})} = e^{J(\Theta -\tilde{\Theta})} J$;

\item $\displaystyle Q_L(e^{J(\Theta -\tilde{\Theta})}u, v)
-e^{J(\Theta -\tilde{\Theta})} Q_L(u, v) = \frac{(\sigma
+1) \sqrt{\omega}}{2\pi} \sin(\Theta -\tilde{\Theta}) \sigma_3
q_u \overline{q_v}$, 
for any $u$, $v \in V$, where
 $$\sigma_3 =\left[
  \begin{array}{cc}
    1 & 0 \\
    0 & -1 \\
  \end{array}
\right].$$
\end{itemize}

\noindent Making use of the previous identities, one rewrites the
  equation for $\eta$ as  
\begin{equation}    \label{eq-da-congelare}
\left( \frac{d\eta}{dt}, v \right)_{L^2} = (\omega_1 -\omega)
(J\eta, v)_{L^2} +Q_L(\eta, v) +\left( e^{-J(\Theta -\tilde{\Theta})}
\left(\dot{\gamma} J\Phi_{\omega} -\dot{\omega}
\frac{d\Phi_{\omega}}{d\omega} \right), v \right)_{L^2}+
\end{equation}
$$+e^{-J(\Theta -\tilde{\Theta})} \frac{(\sigma +1)
\sqrt{\omega}}{2\pi} \sin(\Theta -\tilde{\Theta}) \sigma_3 q_{\eta}
\overline{q_v} +e^{-J(\Theta -\tilde{\Theta})} JN(e^{J(\Theta -\tilde{\Theta})}
q_{\eta}) \overline{q_v},\ \ \forall v\in V\ .$$

\noindent Let us define the linearization frozen at time $t_1$ as $L_I = L(\omega_1),$ and observe that for all $u$, $v \in V$
$$Q_L(u,v) -Q_{L_I}(u,v) =
\frac{\sqrt{\omega} -\sqrt{\omega_1}}{4\pi} \mathbb{T} q_u \overline{q_v}
-(\omega_1 -\omega) (Ju, v)_{L^2},$$

\noindent where $\displaystyle \mathbb{T} =\left[
  \begin{array}{cc}
    0 & -1 \\
    2\sigma +1 & 0 \\
  \end{array}
\right].$
Hence, equation \eqref{eq-da-congelare} becomes
\begin{equation}\label{eq-congelata}
\left( \frac{d\eta}{dt}, v \right)_{L^2} = Q_{L_I}(\eta, v) + N_I(t,\omega,q_\eta,q_v)\ \ \ \forall v\in V\ ,
\end{equation}
where the time dependent nonlinear remainder (including now
``dragging" terms due to reparametrization) is given for all $v \in V$ by 
\begin{align}
N_I(t,\omega,q_\eta,q_v)&=\left(
e^{-J(\Theta -\tilde{\Theta})} \left(\dot{\gamma} J\Phi_{\omega}
-\dot{\omega} \frac{d\Phi_{\omega}}{d\omega}\right), v \right)_{L^2}+\frac{\sqrt{\omega} -\sqrt{\omega_1}}{4\pi} \mathbb{T} q_{\eta}
\overline{q_v}\
\\
&+e^{-J(\Theta -\tilde{\Theta})} \frac{(\sigma +1) \sqrt{\omega}}{2\pi}
\sin(\Theta -\tilde{\Theta}) \sigma_3 q_{\eta} \overline{q_v}+e^{-J(\Theta -\tilde{\Theta})} JN(e^{J(\Theta -\tilde{\Theta})} q_{\eta})
\overline{q_v}\ .
\end{align}

\noindent The gain in changing from original \eqref{eqchi} for the
dispersive component to equation \eqref{eq-congelata} is
that the latter is still non autonomous, but now the
generator of the evolution is (in weak form) a sum of a fixed linear
vector field (the frozen linearization $L_I$) and a nonlinear time
dependent perturbation (see also \cite{BP1}). This allows to use
the known dispersive properties of linearization operator $L$
described in \ref{dispersive}. 

\subsection{Duhamel's representation}	\label{duhamel}
In this subsection we write the equation
\eqref{eq-congelata} in Duhamel's representation to better exploit
the dispersive properties of the propagator $e^{L_It}$. This is not a completely
trivial task since our frozen equation is a variational equation and cannot be
written in strong form. In order to reach our purpose, we consider \eqref{eq-congelata} separating in the test function $v$ the regular and singular part accordingly to \eqref{def-V2}. So we begin by setting
$v = \phi_v^{\lambda} \in H^1(\R^3)$. We get
$$\left( \frac{d\eta}{dt}(t), \phi_v^{\lambda} \right)_{L^2} = (L_I
\eta(t) +f_I(t), \phi_v^{\lambda})_{L^2},$$

\noindent where $f_I(t) = e^{-J(\Theta(t) -\widetilde{\Theta}(t))}
\left( \dot{\gamma}(t) J\Phi_{\omega(t)} -\dot{\omega}(t)
\frac{d\Phi_{\omega(t)}}{d\omega} \right)$. Hence, by Duhamel's
principle one gets
$$(\eta, \phi_v^{\lambda})_{L^2} = \left( e^{L_I t} \eta_0 +\int_0^t
e^{L_I (t-s)} f_I(s) ds, \phi_v^{\lambda} \right)_{L^2}.$$

\noindent If one considers the same equation with $v = q_v
G_{\lambda}$ when $q_v \in \C$, one has
$$\left( \frac{d\eta}{dt}(t), q_v G_{\lambda} \right)_{L^2} = (L_I
\eta(t) +f_I(t) +g_I(t), q_v G_{\lambda})_{L^2},$$

\noindent where
$$g_I(t) = e^{-J(\Theta(t) -\widetilde{\Theta}(t))}
\left( 4\sqrt{\lambda} (\sigma +1) \sqrt{\omega(t)} \sin(\Theta(t)
-\widetilde{\Theta}(t)) \sigma_3 q_{\eta}(t) G_{\lambda} +\right.$$
$$\left. +8\pi \sqrt{\lambda} JN(e^{J (\Theta(t) -\widetilde{\Theta}(t))} q_{\eta}(t)) G_{\lambda}
\right) +2\sqrt{\lambda} (\sqrt{\omega(t)} -\sqrt{\omega_1}) \mathbb{T}
q_{\eta}(t) G_{\lambda},$$
where $q_\eta$ is the charge of the
function $\eta$.
\noindent Hence one has
$$(\eta, q_v G_{\lambda})_{L^2} = \left(
e^{L_I t} \eta_0 +\int_0^t e^{L_I (t-s)} (f_I(s) +g_I(s)) ds, q_v
G_{\lambda} \right)_{L^2}.$$

\noindent Summing up, for any $v \in V$, the equation \eqref{eq-congelata}
can be rewritten as
$$(\eta, v)_{L^2} = \left(
e^{L_I t} \eta_0 +\int_0^t e^{L_I (t-s)} f_I(s) ds, v
\right)_{L^2} +\left( \int_0^t e^{L_I (t-s)} g_I(s) ds, q_v G_{\lambda}
\right)_{L^2}.$$

\noindent In what follows we will use the following estimate on the function
$g_I$.

\begin{lemma}
Under the hypotheses of Proposition \ref{prop-maggioranti}, there
exists a constant $C > 0$ such that
$$\|g_I(t)\|_{V \cap L^1_w} \leq C(|q_\eta|^2 +\rho |q_{\eta}|),$$

\noindent for any $t \leq t_1$.
\end{lemma}

\begin{proof}
First of all let us notice that it is possible to chose $t_1$ in such a way that $\omega (t) \geq c > 0$ for any $0 \leq t \leq t_1$, then
$$|\sqrt{\omega(t)} -\sqrt{\omega_1}| \leq C |\omega(t) -\omega_1| \leq
C \int_t^{t_1} |\dot{\omega}(s)| ds \leq C \sup_{0 \leq t \leq t_1} \left[
(1+t)^3 |\dot{\omega}(t)| \right]
\int_t^{t_1} (1+s)^{-3} ds \leq C \rho,$$

\noindent and
$$|\Theta(t) -\widetilde{\Theta}(t)| \leq \int_0^t \int_s^{t_1}
|\dot{\omega}(\tau)| d\tau ds + \int_t^{t_1} |\dot{\gamma}(s)| ds \leq C\rho
\int_0^t \int_s^{t^1} (1+\tau)^{-3} d\tau ds +C\rho \int_t^{t_1}
(1+s)^{-3} ds \leq C \rho.$$

\noindent The result follows since
$$\|g_I(t)\|_{V \cap L^1_w} \leq C (|\Theta(t) -\widetilde{\Theta}(t)|
|q_{\eta}(t)| +|\sqrt{\omega(t)} -\sqrt{\omega_1}| |q_{\eta}(t)|
+|q_{\eta}(t)|^2).$$

\end{proof}
We end the section with a technical result that allows to transfer
dispersive estimates on the frozen fluctuating component
$P_c(L_I)\eta=P_c(\omega_1)\eta$ into estimates on $\eta$. This is
needed because $\eta$ appears in the integral Duhamel's equation
where estimates have to be done, but the dispersive behaviour is at
our disposal for $P_c(\omega_1)\eta$. 
This is stated in the following lemma (see for analogous construction, for example, \cite{GS1} and \cite{BKKS}).
\begin{lemma}\label{proiezioneeta}
Let the hypotheses of Proposition \ref{prop-maggioranti} hold true
and suppose that the quantity 
$$\sup_{0 \leq t\leq t_1} (|\omega(t) -\omega_1| +|\Theta(t)
-\widetilde{\Theta}(t)|) = \delta$$

\noindent is sufficiently small; then, for any $t \in [0, t_1]$ there
is a bounded linear operator $\Pi(t): P_c(\omega_1)(V \cap
L^{\infty}_{w^{-1}}) \rightarrow V \cap L^{\infty}_{w^{-1}}$, and a
positive constant $C = C(\delta, \omega_1) > 0$ such that $\eta(t) =
\Pi(t) h(t)$, and
$$C(\delta, \omega_1)^{-1} \| h \|_{V \cap L^{\infty}_{w^{-1}}} \leq
\| \eta \|_{V \cap L^{\infty}_{w^{-1}}} \leq C(\delta, \omega_1) \|
h \|_{V \cap L^{\infty}_{w^{-1}}}.$$

\end{lemma}
\begin{proof}
We give only a sketch of the standard proof, referring for details to
the literature cited above.
Set $\eta(t)=P_0(\omega_1)\eta+P_c(\omega_1)\eta=i
k_1(t)\Phi_{\omega_1} +k_2(t)\frac{d}{d\omega_1}\Phi_{\omega_1} +
h(t)\ .$ The condition $P_0\chi=0\ $ makes time dependent functions
$k_1$ and $k_2$ to satisfy a linear system with a source term
depending on $h$; the coefficient matrix has an inverse uniformly
bounded in $t$ and  $t_1$ thanks to the conditions
$(\Phi_{\omega},\frac{d}{d\omega_1}\Phi_{\omega_1})_{L^2}>{\rm
const}>0$ and
$(\Phi_{\omega_1},\frac{d}{d\omega}\Phi_{\omega})_{L^2}>{\rm const}>0$
valid for $|\omega-\omega_1|$ small enough. This gives a
representation of $k_1$ and $k_2$ in terms of $h$ and as a consequence
the required bound on the finite dimensional component. Now define
$\Pi(t)h(t)=\eta(t)-ik_1\Phi_{\omega_1}
-k_2\frac{d}{d\omega_1}\Phi_{\omega_1}\ $ and the complete bound
follows. 
\end{proof}
\noindent 
\subsection{Proof of Proposition \ref{prop-maggioranti}}	\label{proof-maggioranti}

\noindent \texttt{Estimate of $|\dot{\gamma}| +|\dot{\omega}|$.}

\begin{lemma}
If $\eta \in V \cap L^{\infty}_{w^{-1}}$ then its charge $q_{\eta}$ satisfies $|q_{\eta}| \leq 4\pi \| \eta \|_{L^{\infty}_{w^{-1}}}.$

\end{lemma}

\begin{proof}
Since $\eta \in L^{\infty}_{w^{-1}}(\R^3)$ then
$\| \eta \|_{L^{\infty}_{w^{-1}}}
= \sup_{x \in \R^3} \left| \frac{|x|}{1 +|x|} \phi_{\eta}(x)
+\frac{q_{\eta}}{4\pi (1+|x|)} \right| \geq \frac{1}{4 \pi} |q_{\eta}|
.$



\end{proof}

\noindent From the last lemma and Corollary \ref{stime-modulazione}
one gets
$$|\dot{\gamma}(t)| +|\dot{\omega}(t)| \leq c |q_{\eta}(t)|^2 \leq
c_1 \| \eta(t) \|_{L^{\infty}_{w^{-1}}}^2 \leq c_1 (1+t)^{-3} M(t)^2,
\quad \forall t \in [0, t_1],$$

\noindent with $c_1$ independent of $t_1$.
Hence, one can choose ${\rho_1}^2 < \frac{1}{4c_1}$ and get $(1+t)^{3}
(|\dot{\gamma}(t)| +|\dot{\omega}(t)|) \leq c_1 \rho^2 
\leq \frac{\rho}{4}, \quad \forall t \in [0, t_1].$

\noindent \texttt{Estimate of $\| \eta \|_{L^{\infty}_{w^{-1}}}$.}

\noindent As explained in the previous section, for any $t \in [0, t_1]$ we have $\eta(t) =
P_0(\omega_1)\eta(t) +P_c(\omega_1)\eta(t)$ (for the definitions of
$P_0$ and $P_c$ see Remark \ref{split}) and thanks to Lemma \ref{proiezioneeta} we have $\eta(t)=\Pi h(t)$ where $\Pi(t): P_c(\omega_1)(V \cap
L^{\infty}_{w^{-1}}) \rightarrow V \cap L^{\infty}_{w^{-1}}$ is bounded.

In order to estimate $\| \eta \|_{L^{\infty}_{w^{-1}}}$ we make use of the equation for $h$. For all $v \in V$, $h$ is a solution to
$$\left( \frac{dh}{dt}, v \right)_{L^2} = Q_{L_I}(h, v) +\left( P_c(\omega_1)
f_I, v \right)_{L^2} +\left( P_c(\omega_1) g_I, g_v G_{\lambda} \right)_{L^2},$$

\noindent where $f_I$ and $g_I$ were defined at the beginning of Section VI B; hence, for any $v \in V$, $h$ satisfies
$$(h, v)_{L^2} = \left(
e^{L_I t} h_0 +\int_0^t e^{L_I (t-s)} P_c(\omega_1) f_I(s) ds, v
\right)_{L^2} +\left( \int_0^t e^{L_I (t-s)} P_c(\omega_1)
g_I(s) ds, q_v G_{\lambda} \right)_{L^2}.$$

\noindent In addition let us assume that $v \in V \cap L^1_{w}$,
hence by H\"{o}lder inequality
$$(h, v)_{L^2} \leq \left( \|e^{L_I t} h_0\|_{V \cap L^{\infty}_{w^{-1}}}
+\left\|\int_0^t e^{L_I (t-s)} P_c(\omega_1) f_I(s) ds\right\|_{V
\cap L^{\infty}_{w^{-1}}} \right) \|v\|_{L^1_w} +$$
$$+\left\|\int_0^t e^{L_I (t-s)} P_c(\omega_1) g_I(s) ds\right\|_{V \cap
L^{\infty}_{w^{-1}}} \|q_v G_{\lambda}\|_{L^1_w}.$$

\noindent Now we can apply the dispersive estimate proved in Theorem
\ref{stima-cont} and get
$$\|e^{L_I t} h_0\|_{V \cap L^{\infty}_{w^{-1}}} \leq c (1+t)^{-3/2}
\|h_0\|_{V \cap L^1_w} \leq c (1+t)^{-3/2} d\ ,$$
where $d$ was defined in the statement of the
present proposition.
\noindent Furthermore, again by Theorem IV.8,
$$\left\|\int_0^t e^{L_I (t-s)} P_c(\omega_1) f_I(s) ds\right\|_{V \cap
L^{\infty}_{w^{-1}}} \leq c \int_0^t (1+ t -s)^{-3/2} \|f_I(s)\|_{V \cap L^1_w}
ds \leq$$
$$\leq c \int_0^t (1+ t -s)^{-3/2} (|\dot{\gamma}(s)| +|\dot{\omega}(s)|) ds
\leq c \int_0^t (1+ t -s)^{-3/2} \|\eta(s)\|_{L^{\infty}_{w^{-1}}}^2 ds.$$

\noindent Analogously, using Lemma VI.3 and Theorem IV.8,
$$\left\|\int_0^t e^{L_I (t-s)} P_c(\omega_1) g_I(s) ds\right\|_{V
\cap L^{\infty}_{w^{-1}}} \leq c \int_0^t (1+ t -s)^{-3/2} \|g_I(s)\|_{V \cap
L^1_w} ds \leq$$
$$\leq c \int_0^t (1+ t -s)^{-3/2} (\|\eta(s)\|_{L^{\infty}_{w^{-1}}}^2
+\rho \|\eta(s)\|_{L^{\infty}_{w^{-1}}}) ds.$$

\noindent Let us define
$$m(t) = \sup_{s \in [0,t]} (1+s)^{3/2}
\|\eta(s)\|_{L^{\infty}_{w^{-1}}}.$$

Now, using the above inequalities, Lemma 1.25, and exploiting the duality pairing defined by the inner
product in $L^2$, it holds
$$(1 +t)^{3/2} \| \eta (t) \|_{L^{\infty}_{w^{-1}}} = (1 +t)^{3/2}
\sup_{0 \neq v \in L^1_w} \frac{(\eta(t), v)_{L^2}}{\| v \|_{L^1_w}}
\leq$$
$$\leq c \int_0^t (1+ t -s)^{-3/2} (\|\eta(s)\|_{L^{\infty}_{w^{-1}}}^2
+\rho \|\eta(s)\|_{L^{\infty}_{w^{-1}}}) ds$$
$$\leq c \left( d +m^2(t) \int_0^t (1+t)^{3/2} (1+s)^{-3}(1+ t -s)^{-3/2} ds +\rho m(t)
\int_0^t (1+t)^{3/2} (1+s)^{-3/2}(1+ t -s)^{-3/2} ds \right).$$



\noindent Observe that the constant $c$ and both integrals appearing
in the last inequality are bounded independently of $t$, and this implies
that for any $t \in [0, t_1]$ we have
$$m(t) \leq c (d +m^2(t_1) +\rho m(t_1)) \leq c (d +\rho_1^2) \leq c_2 d,$$

\noindent provided $d$ and $\rho$ are small enough. Since the
constant $c_2$ does not depend on $t_1$, we can choose $d <
\frac{\rho}{4c_2}$ and finally get
$$m(t_1) \leq \frac{\rho}{4},$$

\noindent concluding the proof of Proposition
\ref{prop-maggioranti}.

\section{Asymptotic stability} \label{stabasi}
Now we are in the position to prove the asymptotic stability result
as stated in the next theorem. Before formulating the result,
let us denote by $U_t$ the integral kernel which defines the
propagator of the free Laplacian in $\R^3$, namely $\displaystyle
U_t(x) = (4\pi it)^{-3/2} e^{i\frac{|x|^2}{4t}}.$

\begin{theo}
Assume $\sigma \in (0, 1/\sqrt{2})$. Let $u \in C(\R^+, V)$ be a
solution to equation \eqref{eq1} with $u(0) = u_0 \in V \cap L^1_w$
and denote $d = \| u_0 -e^{i\theta_0} \Phi_{\omega_0} \|_{V \cap
L^1_w},$ for some $\omega_0 > 0$ and $\theta_0 \in \R$. Then, if $d$
is sufficiently small, the solution $u(t)$ can be decomposed as
follows
\begin{equation}
u(t) = e^{i\omega_{\infty} t} \Phi_{\omega_{\infty}}
+U_t*\psi_{\infty} +r_{\infty}(t), 
\end{equation}

\noindent where $\omega_{\infty} > 0$ and $\psi_{\infty}$,
$r_{\infty}(t) \in L^2(\R^3)$, with $\| r_{\infty}(t) \|_{L^2} =
O(t^{-5/4})$ as $t \rightarrow +\infty.$

\end{theo}

\begin{proof}
Along the proof we assume that $P_0 (u_0 -e^{i\theta_0}
\Phi_{\omega_0}) = 0$, and we recall from Lemma \ref{dati-iniz} that
there is no loss of generality in this choice. First of all let us
notice that Theorem \ref{teo-maggioranti} implies $\omega(t)
\rightarrow \omega_{\infty},$ and $\Theta(t) -\omega_{\infty}t
\rightarrow 0$, as $t \rightarrow +\infty$.
Next, let us define the modulated soliton as
$$s(t,x) = e^{i\Theta(t)} \Phi_{\omega(t)}(x),$$

\noindent and the function
\begin{equation}    \label{z}
z(t,x) = u(t,x) -s(t,x).
\end{equation}

\noindent By equation \eqref{eq1} and \eqref{gs_eq} 
one has that, for any $v \in V$, $z(t)$ is also a solution to
$$\left( i\frac{dz}{dt}, v \right)_{L^2} = \Re \int_{\R^3} \nabla \phi_z \cdot \overline{\nabla \phi_v} dx -\nu \Re (
(|q_u|^{2\sigma}q_u -|q_s|^{2\sigma}q_s) \overline{q_v}) +\left(
\dot{\gamma}s -i\dot{\omega} \frac{ds}{d\omega}, v \right)_{L^2}.$$

\noindent As one can verify by direct differentiation, the
solution of the last equation can be expressed as
\begin{equation}    \label{z-ansatz}
z(t,x) = U_t*z_0(x) +i\int_0^t U_{t-\tau}(x) q_z(\tau) d\tau
-i\int_0^t U_{t-\tau}*f(s(\tau)) d\tau,
\end{equation}

\noindent where we denoted $f(s) = \dot{\gamma}s -i\dot{\omega}
\frac{ds}{d\omega}$ and, according to \eqref{z}, $q_z(t) = q_u(t)
-q_s(t)$. Let us consider the last integral in formula
\eqref{z-ansatz}
$$\int_0^t U_{t-\tau}*f(s(\tau)) d\tau = U_t*\int_0^{\infty}
U_{-\tau}*f(s(\tau)) d\tau -\int_t^{\infty} U_{t-\tau}*f(s(\tau))
d\tau,$$

\noindent and note that the regularity of $s(t,x)$ implies
$\psi_1(x) = \int_0^{\infty} U_{-\tau}*f(s(\tau)) d\tau \in
L^2(\R^3)$, and $r_1(t,x) = -\int_t^{\infty} U_{t-\tau}*f(s(\tau))
d\tau \in L^2(\R^3)$. Moreover, from Theorem \ref{teo-maggioranti}
and the unitarity of the evolution group of the free Laplacian we
have $\|r_1(t)\|_{L^2} = O(t^{-2})$, $t \rightarrow +\infty.$

\noindent To conclude the proof it is left to prove a similar
asymptotic decomposition for the first integral in the formula
\eqref{z-ansatz}. As before, one can write
$$\int_0^t U_{t-\tau}(x) q_z(\tau) d\tau 
=  U_t * \int_0^{\infty} U_{-\tau}(x)q_z(\tau) d\tau
-\int_t^{\infty} U_{t-\tau}(x)q_z(\tau) d\tau.$$

\noindent First of all one needs to show that $\psi_0(x) =
\int_0^{\infty} U_{-\tau}(x)q_z(\tau) d\tau$ belongs to
$L^2(\R^3)$. To this aim,
let us observe that $\psi_0(x) =
\frac{1}{(4\pi i)^{3/2}} h\left(\frac{r^2}{4}\right)$, with $h(y) =
\int_0^{\infty} e^{-i y/ \tau} \tau^{-3/2} q_z\left( \tau \right)
d\tau$, hence
$$\|\psi_0\|_{L^2}^2 = \frac{1}{(4\pi)^2} \int_0^{\infty} \left| h\left(\frac{r^2}{4}\right)
\right|^2 r^2 dr = \frac{1}{(2\pi)^2} \int_0^{\infty} |h(y)|^2
\sqrt{y} dy.$$

\noindent From the first and the last terms one gets $\psi_0 \in
L^2(\R^3)$ if and only if $h \in L^2(\R^+, \sqrt{y}dy).$ On the
other hand, one can perform the change of variable $u =
\frac{1}{\tau}$ in the integral function $h$ and get
$$h(y) = \int_0^{\infty} e^{-i y u} \frac{1}{\sqrt{u}} q_z\left( \frac{1}{u} \right) du
= \int_0^{\infty} e^{-i y u} \frac{1}{u} q_z\left(
\frac{1}{u} \right) \sqrt{u} du,$$

\noindent where we set $y = \frac{|x|^2}{4}$. Then $\widehat{h}(u) =
\frac{1}{u} q_z\left( \frac{1}{u} \right)$. Moreover, by Theorem
\ref{teo-maggioranti}, $\left| \frac{1}{u} q_z \left( \frac{1}{u}
\right) \right|^2 \sqrt{u} \leq \frac{u^{3/2}}{(1+u)^3}$ then
$\widehat{h} \in L^2(\R^+, \sqrt{u}du)$ and hence, by
Plancherel's identity $h \in L^2(\R^+, \sqrt{y}dy)$.

\noindent Finally, let us denote $r_0 = \int_t^{\infty}
U_{t-\tau}(x)q_z(\tau) d\tau$. As before, we have $r_0(x) = g \left(
\frac{r^2}{4} \right)$, with $g(y) = \int_0^{\infty} e^{-i y/(t-
\tau)} (t-\tau)^{-3/2} q_z\left( \tau \right) d\tau$. Moreover, we
can set $y = \frac{|x|^2}{4}$ exploit the change of variables $u =
-\frac{1}{t-\tau}$ in order to get
$$g(y) = \int_0^{\infty} e^{-iyu} \frac{i}{u} q_z\left( t+\frac{1}{u}
\right) \sqrt{u} du.$$

\noindent Again, Theorem \ref{teo-maggioranti} implies that
$\widehat{g}(u) = \frac{i}{u} q_z\left( t+\frac{1}{u} \right) \in
L^2(\R^+, \sqrt{u} du)$, for any $t \geq 0$. In particular,
$$\|g\|^2_{L^2(\R^+, \sqrt{u} du)} \leq \tilde{c} \int_0^{\infty}
\frac{u^{3/2}}{((1+t)u+1)^3} du \leq c (1+t)^{-5/2},$$

\noindent for any $t \geq 0$, with $\tilde{c}$, $c > 0$ independent
of time. Summing up, Plancherel's identity allows us to conclude
$\|r_0\|_{L^2} = O(t^{-5/4})$ as $t \rightarrow +\infty.$

\noindent Hence the theorem follows with $\psi_{\infty} = z_0
+\psi_0 +\psi_1,$ and $r_{\infty} = r_0 +r_1.$
\end{proof}

\appendix
\section{The generalized kernel of the operator $L$} \label{nucleo}
The aim of this appendix is to provide the proof or Theorem \ref{kergen}.
\begin{proof}
It is easy to see that $c\Phi_{\omega}$, with $c \in
\C$, is the unique family of distributional solutions to the
equation
$$-\triangle u +\omega u = 0.$$

\noindent Furthermore, $\Phi_{\omega}$ belongs to
$D(H_{\alpha_2})$ but not to $D(H_{\alpha_1})$ since the boundary
condition is not satisfied. Hence
$$\ker(L) = \textrm{span} \left\{ \left(
                   \begin{array}{ll}
                     0\\
                     \Phi_{\omega}
                   \end{array}
                 \right)
          \right\}.$$

\noindent Let us now consider the operator
$$L^2 = \left[
  \begin{array}{cc}
    -L_2 L_1 & 0 \\
    0 & -L_1 L_2 \\
  \end{array}
\right].$$

\noindent Since the operator $L_1$ is invertible, the following
holds
$$u \in \ker(L_1 L_2) \Leftrightarrow u \in \ker(L_2), \ \ {\mbox{then}}
\ \
\ker(L_1 L_2) = {\mbox{span}} \{ \Phi_{\omega} \},$$
$$u \in \ker(L_2 L_1)
\Leftrightarrow \exists u \in D(H_{\alpha_1}) \quad
\textrm{such that } \ L_1 u = \Phi_{\omega}.$$

\noindent Solving the former equation one gets that $\ker(L_1 L_2) =
$ span $\{ \varphi_{\omega} \}$. From this follows
$$\ker(L^2) = \textrm{span} \left\{ \left(
                   \begin{array}{ll}
                     0\\
                     \Phi_{\omega}
                   \end{array}
                 \right),
                    \left(
                   \begin{array}{ll}
                     \varphi_{\omega}\\
                     0
                   \end{array}
                 \right)
          \right\}.$$

\noindent The operator $L^3$ has the following form
$$L^3 = \left[
  \begin{array}{cc}
    0 & -L_2 L_1 L_2 \\
    L_1 L_2 L_1 & 0 \\
  \end{array}
\right].$$

\noindent As before
$$u \in \ker(L_1 L_2 L_1) \Leftrightarrow L_1 u \in \ker(L_1 L_2) =
  {\textrm{span}} \ \{ \Phi_{\omega} \} \Leftrightarrow \ker(L_1 L_2 L_1) =
   {\textrm{span}} \ \{
\varphi_{\omega} \},$$
$$u \in \ker(L_2 L_1 L_2) \Leftrightarrow u \in \ker(L_2) =
       {\textrm{span}} \ \{
\Phi_{\omega} \} \quad \textrm{or} \quad L_2 u \in \ker(L_2 L_1) =
  {\textrm{span}} \ \{ \varphi_{\omega} \}.$$

\noindent Let us notice that the equation
$$-\triangle u + \omega u = \varphi_{\omega}$$

\noindent has a unique family of distributional solutions given by
$$u(x) = \left( \frac{\sqrt{\omega}}{4\pi \nu}
\right)^{\frac{1}{2\sigma}} \left[ \left( \frac{c_2}{2\sqrt{\omega}}
+\frac{1}{16\sigma^2 \omega^2} \right) \frac{e^{-\sqrt{\omega}
|x|}}{4\pi |x|} +\frac{c_1}{2\sqrt{\omega}} \frac{e^{\sqrt{\omega}
|x|}}{4\pi |x|} +\right.$$
$$\left. -\frac{1}{8\omega} |x| \frac{e^{-\sqrt{\omega} |x|}}{4\pi}
+\left( \frac{1}{8\sigma \omega^{\frac{3}{2}}}
-\frac{1}{8\omega^{\frac{3}{2}}} \right) \frac{e^{-\sqrt{\omega}
|x|}}{4\pi} \right].$$

\noindent Notice that one must impose that $u$ belongs to
$D(H_{\alpha_2})$ which
means that $u \in L^2(\R^3)$ and satisfies the boundary condition.
This is equivalent to ask the following algebraic conditions to be
verified
$$\left\{
  \begin{array}{ll}
    c_1 = 0\\
    c_2 = \frac{\sigma -1}{8\sigma \omega^{\frac{3}{2}}}.
  \end{array}
\right.$$

\noindent Therefore, if $\sigma \neq 1$, then $\ker(L_2 L_1 L_2) =
          {\textrm{span}} \ \{
\Phi_{\omega} \}$. Hence
$$\ker(L^3) = \ker(L^2),$$

\noindent which concludes the first part of the theorem.

\noindent In the case $\sigma = 1$ we get $\ker (L_2 L_1 L_2)
= \{ \Phi_{\omega}, g_{\omega} \}$, then
$$\ker(L^3) = \textrm{span} \left\{ \left(
                   \begin{array}{ll}
                     0\\
                     \Phi_{\omega}
                   \end{array}
                 \right),
                \left(
                   \begin{array}{ll}
                     \varphi_{\omega}\\
                     0
                   \end{array}
                 \right),
                \left(
                   \begin{array}{ll}
                     0\\
                     g_{\omega}
                   \end{array}
                 \right)
 \right\}.$$

\noindent With analogous computations one can prove that
$$\ker(L^4) = \ker(L^5) = \textrm{span} \left\{ \left(
                   \begin{array}{ll}
                     0\\
                     \Phi_{\omega}
                   \end{array}
                 \right),
                \left(
                   \begin{array}{ll}
                     \varphi_{\omega}\\
                     0
                   \end{array}
                 \right),
                \left(
                   \begin{array}{ll}
                     0\\
                     g_{\omega}
                   \end{array}
                \right),
                \left(
                   \begin{array}{ll}
                     h_{\omega}\\
                     0
                   \end{array}
                 \right)
 \right\},$$

\noindent which concludes the proof.
\end{proof}

\section{Proof of the resolvent formula} \label{risolvente}
In this appendix we prove that the operator $(L -\lambda I)^{-1}$ is
given by
$$R(\lambda) = \left[
  \begin{array}{cc}
    -\lambda (\lambda^2 +L_2 L_1)^{-1} & -L_2 (\lambda^2 +L_1 L_2)^{-1}\\
    L_1 (\lambda^2 +L_2 L_1)^{-1} & -\lambda (\lambda^2 +L_1 L_2)^{-1}\\
  \end{array}
\right]$$

\noindent for the resolvent of the linear operator $L$. More precisely, we
prove the following proposition.

\begin{prop}
If $\lambda \in \C \setminus \sigma(L)$, then
$R(\lambda) (L -\lambda I) u = u, \quad \forall u \in D(L),$
and
$(L -\lambda I) R(\lambda) f = f$ for  any $f \in (L^2(\R^3))^2.$

\end{prop}
\noindent Before proving the former proposition, let us prove the following
lemma.

\begin{lemma} \label{lemma-commut}
For any $\lambda \in \C \setminus \sigma(L)$ the following identities hold

\begin{enumerate}
\item $(\lambda ^2 +L_2 L_1)^{-1} L_1^{-1} = L_1^{-1} (\lambda ^2 +L_1
L_2)^{-1}$, \label{lemma1}
\item $(\lambda ^2 +L_1 L_2)^{-1} = (\lambda^2 L_1^{-1} +L_2)^{-1} L_1^{-1}$,
\label{lemma2}
\item $(\lambda^2 +L_1 \widetilde{L_2})^{-1} \widetilde{L_2}^{-1} =
\widetilde{L_2}^{-1} (\lambda^2 +\widetilde{L_2} L_1)^{-1}$, \label{lemma3}
\end{enumerate}

\noindent where $\widetilde{L_2}$ is the restriction of the operator $L_2$ to
the projection of its domain onto the subspace of $L^2(\R^3)$
associated to the continuous spectrum of $L_2$.
\end{lemma}

\begin{proof}
First of all, let us notice that all the inverse operators are well
defined since
$\lambda$ is
not allowed to be a spectral point of $L$, $L_1$ is invertible and $L_2$ is
restricted to a subspace on which it is invertible too.

\noindent In order to prove \ref{lemma1}, we prove the following claim
$$(\lambda ^2 +L_2 L_1)^{-1} L_1^{-1} = (\lambda^2 L_1 +L_1 L_2 L_1)^{-1} =
L_1^{-1} (\lambda ^2 +L_1
L_2)^{-1}.$$

\noindent To this purpose, let us take any $\xi \in L^2(\R^3)$, then one has
$$(\lambda ^2 +L_2 L_1)^{-1} L_1^{-1} \xi \in D(L_2 L_1) \quad \textrm{and}
\quad L^{-1}_1 \xi \in D(L_1).$$

\noindent Hence, the following chain of identities holds
$$(\lambda^2 L_1 +L_1 L_2 L_1) (\lambda ^2 +L_2 L_1)^{-1} L_1^{-1} \xi = L_1
(\lambda^2 +L_2 L_1) (\lambda^2 +L_2 L_1)^{-1} L_1^{-1}\xi = L_1 L_1^{-1} \xi =
\xi.$$

\noindent On the other hand, let us take $\eta \in D(L_1 L_2 L_1)$, and observe
that, in particular, $\eta \in D(L_2 L_1)$. This justifies the following
identities
$$(\lambda ^2 +L_2 L_1)^{-1} L_1^{-1} (\lambda^2 L_1 +L_1 L_2 L_1) \eta =$$
$$= (\lambda ^2 +L_2 L_1)^{-1} L_1^{-1} L_1 (\lambda^2 +L_2 L_1) \eta = (\lambda
^2 +L_2 L_1)^{-1} (\lambda^2 +L_2 L_1) \eta = \eta,$$

\noindent which concludes the proof of the first identity of the claim. The
second one is proved in the same way.

\noindent The proof of \ref{lemma3}. can be done in the same way exganging $L_1$
with $\widetilde{L_2}$ and $L_2$ with $L_1$.

\noindent It is left to prove \ref{lemma2}.. To do that, let $\xi$ be in
$L^2(\R^3)$, then $(\lambda^2 L_1^{-1} +L_2)^{-1} L_1^{-1} \xi \in D((\lambda^2
L_1^{-1} +L_2))$ and $L_1^{-1} \xi \in D(L_1)$. Hence, we have
$$(\lambda^2 +L_1 L_2) (\lambda^2 L_1^{-1} +L_2)^{-1} L_1^{-1} \xi = L_1
(\lambda^2 L_1^{-1} +L_2) (\lambda^2 L_1^{-1} +L_2)^{-1} L_1^{-1} \xi = \xi.$$

\noindent On the other hand, for any $\eta \in D(L_1 L_2)$ one has $\eta \in
D(L_2) \subset L^2(\R^3) = D(L_1^{-1})$, which justifies
$$(\lambda^2 L_1^{-1} +L_2)^{-1} L_1^{-1} (\lambda^2 +L_1 L_2) \eta = (\lambda^2
L_1^{-1} +L_2)^{-1} L_1^{-1} L_1 (\lambda^2 L_1^{-1} +L_2) \eta = \eta.$$

\end{proof}

\noindent We can now prove the proposition.

\begin{proof}
\texttt{I step: proof of the first identity.}

\noindent Let us recall that for $u \in D(L)$ holds
$$R(\lambda) (L -\lambda I) u =$$
$$= \left[
  \begin{array}{cc}
    -\lambda (\lambda^2 +L_2 L_1)^{-1} & -L_2 (\lambda^2 +L_1 L_2)^{-1}\\
    L_1 (\lambda^2 +L_2 L_1)^{-1} & -\lambda (\lambda^2 +L_1 L_2)^{-1}\\
  \end{array}
\right]
\left[
  \begin{array}{cc}
    -\lambda & L_2\\
    -L_1& -\lambda\\
  \end{array}
\right]
\left(
  \begin{array}{ll}
  u_1\\
  u_2
  \end{array}
\right) =\left(
  \begin{array}{ll}
  w_1\\
  w_2
  \end{array}
\right),$$

\noindent where
$$w_1 = \lambda^2 (\lambda^2 +L_2 L_1)^{-1} u_1 + L_2 (\lambda^2 +L_1 L_2)^{-1}
L_1 u_1 -\lambda (\lambda^2 +L_2 L_1)^{-1} L_2 u_2 +\lambda L_2 (\lambda^2
+L_1 L_2)^{-1} u_2,$$

\noindent and
$$w_2 = \lambda^2 (\lambda^2 +L_1 L_2)^{-1} u_2 + L_1 (\lambda^2 +L_2 L_1)^{-1}
L_2 u_2 +\lambda (\lambda^2 +L_1 L_2)^{-1} L_1 u_1 -\lambda L_1 (\lambda^2
+L_2 L_1)^{-1} u_1.$$

\noindent We will concentrate on the first component $w_1$, because the second
one can be treated in the same way.

\noindent The spectrum of the selfadjoint operator $L_2$
is (\cite{Albeverio})
$$\sigma(L_2) = \{ 0 \} \cup [\omega, +\infty),$$

\noindent where $0$ is a simple eigenvalue and $\ker(L_2) = \textrm{span}\{
\Phi_{\omega} \}$. Hence, any $u_2 \in D(L_2)$ can be decomposed as
$$u_2 = a \Phi_{\omega} +g_2,$$

\noindent where $a \in \C$ and $g_2$ belongs to the projection of $D(L_2)$
onto the continuous spectrum of $L_2$.

\noindent Moreover, since $L_2 \Phi_{\omega} = 0$, one gets $\Phi_{\omega} \in
D(L_1 L_2)$ and
$$\Phi_{\omega} = \frac{1}{\lambda^2} (\lambda^2 +L_1
L_2) \Phi_{\omega} = (\lambda^2 +L_1
L_2) \left( \frac{1}{\lambda^2} \Phi_{\omega} \right),$$

\noindent which is equivalent to
$(\lambda^2 +L_1 L_2)^{-1} \Phi_{\omega} \in \ker(L_2).$

\noindent As a consequence, since $L_1$ and $\widetilde{L_2}$ are invertible on
their domains, one has
$$w_1 = \lambda^2 (\lambda^2 +L_2 L_1)^{-1} L_1^{-1} L_1 u_1 + L_2 (\lambda^2
+L_1 L_2)^{-1} L_1 u_1 +$$
$$-\lambda (\lambda^2 +\widetilde{L_2} L_1)^{-1} \widetilde{L_2} g_2 +\lambda
\widetilde{L_2} (\lambda^2 +L_1 \widetilde{L_2})^{-1} \widetilde{L_2}^{-1}
\widetilde{L_2} g_2,$$

\noindent hence, by lemma \ref{lemma-commut} it follows
$$w_1 = (\lambda^2 L_1^{-1} +L_2) (\lambda^2 +L_1 L_2)^{-1} L_1 u_1 -\lambda
(\lambda^2 +\widetilde{L_2} L_1)^{-1} \widetilde{L_2} g_2 +\lambda
\widetilde{L_2} \widetilde{L_2}^{-1} (\lambda^2 +\widetilde{L_2} L_1)^{-1}
\widetilde{L_2} g_2 =$$
$$= (\lambda^2 L_1^{-1} +L_2) (\lambda^2 L_1^{-1} +L_2)^{-1} L_1^{-1} L_1 u_1
= u_1.$$

\noindent Summing up, we proved
$$R(\lambda) (L -\lambda I) u = u \quad \forall u \in D(L).$$

\noindent \texttt{II step: proof of the second identity.}

\noindent First of all let us recall that for $f \in (L^2(\R^3))^2$ one has
$$(\lambda^2 +L_2 L_1)^{-1} f_1 \in D(L_2 L_1) \quad \textrm{and} \quad
(\lambda^2 +L_1 L_2)^{-1} f_2 \in D(L_1 L_2).$$

\noindent Hence, the following identities hold
$$(L -\lambda I) R(\lambda) f
= \left[
  \begin{array}{cc}
    -\lambda & L_2\\
    -L_1& -\lambda\\
  \end{array}
\right]
\left[
  \begin{array}{cc}
    -\lambda (\lambda^2 +L_2 L_1)^{-1} & -L_2 (\lambda^2 +L_1 L_2)^{-1}\\
    L_1 (\lambda^2 +L_2 L_1)^{-1} & -\lambda (\lambda^2 +L_1 L_2)^{-1}\\
  \end{array}
\right]
\left(
  \begin{array}{ll}
  f_1\\
  f_2
  \end{array}
\right) =$$
$$=\left(
  \begin{array}{ll}
  (\lambda^2 +L_2 L_1) (\lambda^2 +L_2 L_1)^{-1} f_1\\
  (\lambda^2 +L_1 L_2) (\lambda^2 +L_1 L_2)^{-1} f_2
  \end{array}
\right) = f,$$

\noindent which concludes the proof.
\end{proof}

\section{The dynamics generated by $L$ along the generalized
kernel} \label{polo}
\label{propagator}\noindent In this appendix we estimate the behaviour
of the propagator of $L$ along
the eigenvalue $0$. This is achieved in the following theorem in which it is
proved that the dynamics has a linear growth in time along the
generalized kernel.
\begin{theo}    \label{stima-disc}
For any $r \in (0, \omega)$ the following identity holds
$$\frac{1}{2\pi i} \int_{|\lambda| = r} R(\lambda; x, y) e^{-\lambda
t} d\lambda =$$
$$= \left[
  \begin{array}{cc}
    \frac{\sqrt{\omega}}{1 -\sigma} \frac{e^{-\sqrt{\omega} (|x| +|y|)}}{2\pi
|x| |y|} (2\sigma \sqrt{\omega} |x| -1) &
    0\\
    i\frac{2\omega^{\frac{3}{2}} \sigma}{1 -\sigma} \frac{e^{-\sqrt{\omega} (|x|
+|y|)}}{\pi |x| |y|} t & \frac{\sqrt{\omega}}{1 -\sigma} \frac{e^{-\sqrt{\omega}
(|x| +|y|)}}{2\pi |x| |y|} (2\sigma \sqrt{\omega} |y|
    -1)\\
  \end{array}
\right],$$

\noindent for any $x,y \in \R^3$.
\end{theo}

\begin{proof}
Since the convolution term of the resolvent $R(\lambda)$ is
continuous in zero it suffices to compute the integral of the
multiplication term.
First of all, let us note that the function
$$f(\lambda) = \frac{4\pi i}{W(\lambda^2)} \Lambda_1(\lambda)
e^{-\lambda t} = i\frac{e^{-\lambda t}}{W(\lambda^2)} \cdot$$
$$\cdot \left[ \frac{(4\pi \alpha_2 -i\sqrt{-\omega +i\lambda})
e^{i\sqrt{-\omega -i\lambda} |x|} +(4\pi \alpha_2 -i\sqrt{-\omega
-i\lambda}) e^{i\sqrt{-\omega +i\lambda} |x|}}{8\pi |x| |y|}  \left(
e^{i\sqrt{-\omega +i\lambda} |y|} -e^{i\sqrt{-\omega -i\lambda} |y|} \right)
+\right.$$
$$\left. +\frac{-(4\pi \alpha_1 -i\sqrt{-\omega +i\lambda})
e^{i\sqrt{-\omega -i\lambda} |x|} +(4\pi \alpha_1 -i\sqrt{-\omega
-i\lambda}) e^{i\sqrt{-\omega +i\lambda} |x|}}{8\pi |x| |y|} \left(
e^{i\sqrt{-\omega +i\lambda} |y|} +e^{i\sqrt{-\omega -i\lambda} |y|} \right)
\right] =$$
$$= \frac{i}{8\pi |x| |y|} \left[ 2(4\pi \alpha_2 +\sqrt{\omega})
e^{-\sqrt{\omega} |x|} \left( \frac{i|y|}{\sqrt{\omega}}
e^{-\sqrt{\omega} |y|} \lambda +o(\lambda^2) \right) +8\pi \alpha_1
e^{-\sqrt{\omega} |y|} \left(
\frac{i|x|}{\sqrt{\omega}} e^{-\sqrt{\omega} |x|} \lambda
+o(\lambda^2) \right) +\right.$$
$$\left. +2ie^{-\sqrt{\omega} |y|} \left( \left(
\frac{1}{\sqrt{\omega}} +|x| \right) e^{-\sqrt{\omega} |x|} \lambda
+o(\lambda^2) \right) \right] \left( \frac{1 -\sigma}{2\omega}
\lambda^2 +o(\lambda^4) \right)^{-1} \sim$$
$$\sim -\frac{\sqrt{\omega}}{1 -\sigma} \frac{e^{-\sqrt{\omega} (|x|
+|y|)}}{2\pi |x| |y|} [(4\pi \alpha_2 +\sqrt{\omega}) |y| +(4\pi
\alpha_1 +\sqrt{\omega}) |x| +1] \frac{1}{\lambda}.$$

\noindent as $\lambda \rightarrow 0$. Hence the function
$f(\lambda)$ has a pole of order one in zero. Then, by the Cauchy
theorem one gets
$$\frac{1}{2\pi i} \int_{|\lambda| = r} \frac{4\pi i}{W(\lambda^2)}
\Lambda_1(\lambda) e^{-\lambda t} d\lambda =
 -\frac{\sqrt{\omega}}{1 -\sigma} \frac{e^{-\sqrt{\omega} (|x|
+|y|)}}{2\pi |x| |y|} [(4\pi \alpha_2 +\sqrt{\omega}) |y| +(4\pi
\alpha_1 +\sqrt{\omega}) |x| +1] =$$
$$= -\frac{\sqrt{\omega}}{1 -\sigma} \frac{e^{-\sqrt{\omega} (|x|
+|y|)}}{2\pi |x| |y|} [-2 \sigma \sqrt{\omega} |x| +1].$$

\noindent Switching $\alpha_1$ to $\alpha_2$ and vice versa, it
follows
$$\frac{1}{2\pi i} \int_{|\lambda| = r} \frac{4\pi i}{W(\lambda^2)}
\Lambda_2(\lambda) e^{-\lambda t} d\lambda
= -\frac{\sqrt{\omega}}{1 -\sigma} \frac{e^{-\sqrt{\omega} (|x|
+|y|)}}{2\pi |x| |y|} [(4\pi \alpha_1 +\sqrt{\omega}) |y| +(4\pi
\alpha_2 +\sqrt{\omega}) |x| +1] =$$
$$= -\frac{\sqrt{\omega}}{1 -\sigma} \frac{e^{-\sqrt{\omega} (|x|
+|y|)}}{2\pi |x| |y|} [-2 \sigma \sqrt{\omega} |y| +1].$$

\noindent On the other hand, the function
$$\frac{4\pi i}{W(\lambda^2)} \Sigma_1(\lambda) e^{-\lambda t}$$

\noindent is the sum of a continuous function and a function with a
pole of second order in zero, namely
$$g(\lambda) e^{-\lambda t} =$$
$$= i \frac{(4\pi \alpha_1 -i\sqrt{-\omega +i\lambda})
e^{i\sqrt{-\omega -i\lambda} |x|} +(8\pi \alpha_1 -i\sqrt{-\omega
-i\lambda}) e^{i\sqrt{-\omega +i\lambda} |x|}}{W(\lambda^2) 4\pi |x|
|y|} (e^{i\sqrt{-\omega +i\lambda} |y|} +e^{i\sqrt{-\omega -i\lambda} |y|})
e^{-\lambda t}.$$

\noindent Note that $g(\lambda) = \sum_{k=2}^{+\infty} a_k \lambda_k$ with
$$
a_{-2} = i\frac{\omega}{1 -\sigma} \frac{4\pi \alpha_1
+\sqrt{\omega}}{\pi |x| |y|} e^{-\sqrt{\omega} (|x| +|y|)},\quad a_{-1} = 0,
$$
then, by residue theorem,
$$\frac{1}{2\pi i} \int_{|\lambda| = r} \frac{4\pi i}{W(\lambda^2)}
\Sigma_1(\lambda) e^{-\lambda t} d\lambda = - i\frac{\omega}{1 -\sigma}
\frac{4\pi \alpha_1
+\sqrt{\omega}}{\pi |x| |y|} e^{-\sqrt{\omega} (|x| +|y|)} t = i\frac{2 \sigma
\omega^{\frac{3}{2}}}{(1 -\sigma) \pi |x| |y|}
e^{-\sqrt{\omega} (|x| +|y|)} t.$$

\noindent In the same way
$$\frac{1}{2\pi i} \int_{|\lambda| = r} \frac{4\pi i}{W(\lambda^2)}
\Sigma_2(\lambda) e^{-\lambda t} d\lambda = 0,$$

\noindent which concludes the proof.
\end{proof}


\begin{thebibliography}{10}

\bibitem{ADFT}
R.~Adami, G.~Dell'{A}ntonio, R.~Figari, and A.~Teta.
\newblock The {C}auchy problem for the {S}chr\"{o}dinger equation in dimension
  three with concentrated nonlinearity.
\newblock {\em Ann. I. H. Poincar\'{e}}, 20:477--500, 2003.

\bibitem{ADFT2}
R.~Adami, G.~Dell'{A}ntonio, R.~Figari, and A.~Teta.
\newblock Blow-up solutions for the {S}chr\"{o}dinger equation in dimension
  three with a concentrated nonlinearity.
\newblock {\em Ann. I. H. Poincar\'{e}}, 21:121--137, 2004.

\bibitem{Albeverio}
S.~Albeverio, F.~Gesztesy, R.~H\"{o}gh-{K}rohn, and H.~Holden.
\newblock {\em Solvable models in quantum mechanics}.
\newblock 2005.

\bibitem{BKKS}
V.S. Buslaev, A.I. Komech, A.E. Kopylova, and D.~Stuart.
\newblock On asymptotic stability of solitary waves in {S}chr\"{o}dinger
  equation coupled to nonlinear oscillator.
\newblock {\em Communications in partial differential equations}, 33:669--705,
  2008.

\bibitem{BP1}
V.S. Buslaev and G.~Perelman.
\newblock Scattering for the nonlinear {S}chr\"{o}dinger equation: states close
  to a soliton.
\newblock {\em St.Petersbourg Math J.}, 4:1111--1142, 1993.

\bibitem{BP2}
V.S. Buslaev and G.~Perelman.
\newblock On the stability of solitary waves for nonlinear {S}chr\"odinger
  equations.
\newblock {\em Amer.Math.Soc.Transl.}, 164(2):75--98, 1995.

\bibitem{BS}
V.S. Buslaev and C.~Sulem.
\newblock On asymptotic stability of solitary waves for nonlinear
  {S}chr\"{o}dinger equation.
\newblock {\em Ann. I. H. Poincar\'{e}}, 20:419--475, 2003.

\bibitem{Cu}
S.~Cuccagna.
\newblock Stabilization of solution to nonlinear {S}chr\"odinger equations.
\newblock {\em Comm.Pure App.Math.}, 54:1110--1145, 2001.
\newblock erratum {\it ibid.} {\bf 58}, 147 (2005).

\bibitem{CM}
S.~Cuccagna and T.~Mizumachi.
\newblock On asymptotic stability in energy space of ground states for
  nonlinear {S}chr\"odinger equations.
\newblock {\em Comm.Math.Phys.}, 284:51--87, 2008.

\bibitem{DPT}
P.~D'{A}ncona, V.~Pierfelice, and A.~Teta.
\newblock Dispersive estimate for the {S}chroedinger equation with point
  interaction.
\newblock {\em Math. Meth. in Appl. Sci.}, 29:309--323, 2006.

\bibitem{NP}
D.Noja and A.~Posilicano.
\newblock Wave equations with concentrated nonlinearities.
\newblock {\em J.Phys.A:Math.Gen.}, 38:5011--5022, 2005.

\bibitem{DM}
N.~Dorr and B.A. Malomed.
\newblock Soliton supported by localized nonlinearities in periodic media.
\newblock {\em Phys.Rev. A}, 83:033828--1, 033828--21, 2011.

\bibitem{KM}
E.Kirr and \"O.Mizrak.
\newblock Asymptotic stability of ground states in 3d nonlinear schr\"odinger
  equation including subcritical cases.
\newblock {\em Journal of functional analysis}, 257:3691--3747, 2009.

\bibitem{FW}
G.~Fibich and X.P. Wang.
\newblock Stability of solitary waves for nonlinear {S}chr\"odinger equation
  with inhomogeneous nonlinearities.
\newblock {\em Physica D}, 175:96--108, 2003.

\bibitem{FGJS1}
J.~Fr\"ohlich, S.~Gustafson, B.L.G. Jonsson, and I.M. Sigal.
\newblock Solitary wave dynamics in an {E}xternal {P}otential.
\newblock {\em Commun.Math.Phys.}, 250:613--642, 2004.

\bibitem{GS1}
Z.~Gang and I.M. Sigal.
\newblock Asymptotic stability of nonlinear schr\"odinger equations with
  potentials.
\newblock {\em Rev. Math. Phys.}, 17:1143--1207, 2005.

\bibitem{GS2}
Z.~Gang and I.M. Sigal.
\newblock Relaxation of solitons in nonlinear schr\"odinger equations with
  potentials.
\newblock {\em Adv. Math.}, 216:443--490, 2007.

\bibitem{GeS}
F.~Genoud and C.A. Stuart.
\newblock Schr\"odinger equations with a spatially decaying nonlinearity:
  existence and stability of standing waves.
\newblock {\em DCDS}, 21:137--186, 2008.

\bibitem{tavole}
I.S. Gradshteyn and I.M. Ryzhik.
\newblock {\em Tables of integrals, series and products}.
\newblock 1965.

\bibitem{GSS}
M.~Grillakis, J.~Shatah, and W.~Strauss.
\newblock Stabity theory of solitary wawes in the presence of symmetry {I}.
\newblock {\em Journal of functional analysis}, 94:308--348, 1987.

\bibitem{GNT}
S.~Gustafson, K.~Nakanishi, and T.P. Tsai.
\newblock Asymptotic stability and completeness in the energy space for
  nonlinear schr\"odinger equations with small solitary waves.
\newblock {\em Int.Math.Res.Not.}, 66:3559--3584, 2004.

\bibitem{KKS}
A.I. Komech, E.A. Kopylova, and D.~Stuart.
\newblock On asymptotic stability of solitary waves for {S}chr\"{o}dinger
  equation coupled to nonlinear oscillator, {I}{I}.
\newblock {\em Comm. Pure Appl. Anal.}, to appear.

\bibitem{MA}
B.A. Malomed and M.Y. Azbel.
\newblock Modulational instability of a wave scattered by a nonlinear centre.
\newblock {\em Phys.Rev. B}, 47:10402--10406, 1993.

\bibitem{RS4}
M.Reed and B.Simon.
\newblock {\em Methods of modern mathematical physics. {IV: Analysis of
  operators}}.
\newblock Academic Press, 1977.

\bibitem{SW1}
A.~Soffer and M.~Weinstein.
\newblock Multichannel nonlinear scattering for nonintegrable equations.
\newblock {\em Comm.Math.Phys.}, 133:119--146, 1990.

\bibitem{SW2}
A.~Soffer and M.~Weinstein.
\newblock Multichannel nonlinear scattering for nonintegrable equations ii. the
  case of anisotropic potentials and data.
\newblock {\em J.Diff.Eq.}, 98:376--390, 1992.

\bibitem{SKBRC}
A.A. Sukhorukov, Y.S. Kivshar, O.~Bang, J.J. Rasmussen, and P.L. Christiansen.
\newblock Nonlinearity and disorder: Classification and stability of nonlinear
  impurity modes.
\newblock {\em Phys.Rev. E}, 63:036601--18, 2001.

\bibitem{TY1}
T.P. Tsai and H.T. Yau.
\newblock Asymptotic dynamics of nonlinear {S}chr\"odinger equations:
  resonance-dominated and dispersion-dominated solutions.
\newblock {\em Comm.Pure.Appl.Math}, 55:153--216, 2002.

\bibitem{TY2}
T.P. Tsai and H.T. Yau.
\newblock Relaxation of excited states in nonlinear {S}chr\"odinger equations.
\newblock {\em Int.Math.Res.Not.}, 31:1629--1673, 2002.

\bibitem{W1}
M.~Weinstein.
\newblock Lyapunov stability of ground states of nonlinear dispersive evolution
  equations.
\newblock {\em Comm.Pure.Appl.Math}, 39:51--68, 1986.

\end{thebibliography}

\end{document}